\documentclass[11pt, letterpaper]{article}
\usepackage[utf8]{inputenc}
\usepackage{fullpage,times}

\usepackage{amsmath}
\usepackage{amsthm}
\usepackage{amssymb}
\usepackage{amsfonts}
\usepackage{mathrsfs}
\usepackage{graphicx}
\usepackage{setspace}
\usepackage{hyperref}
\usepackage{multicol}
\usepackage{caption}
\usepackage{subcaption}
\usepackage{algorithm}
\usepackage[noend]{algorithmic}
\usepackage[dvipsnames]{xcolor}
\usepackage{enumitem}
\def\dm{d_{\text{min}}}

\newtheorem{theorem}{Theorem}

\newtheorem{lemma}[theorem]{Lemma}
\newtheorem{claim}[theorem]{Claim}

\newtheorem{corollary}[theorem]{Corollary}
\newtheorem{proposition}[theorem]{Proposition}

\newenvironment{reminder}[1]{\smallskip

\noindent {\bf Reminder of #1 }\em}{\smallskip}

\begin{document}

\title{Approximation Algorithms for Min-Distance Problems in DAGs} %TODO Please add

%\titlerunning{Dummy short title} %TODO optional, please use if title is longer than one line

\title{Approximation Algorithms for Min-Distance Problems in DAGs}

\author{Mina Dalirrooyfard\thanks{\url{minad@mit.edu}. Massachusetts Institute of Technology.} \quad \quad \quad \quad Jenny  Kaufmann\thanks{\url{jkaufmann@math.harvard.edu}. Harvard University.}}

\date{\today\footnote{Corrected version.}}

\maketitle

%TODO mandatory: add short abstract of the document
\begin{abstract}
Graph parameters such as the diameter, radius, and vertex eccentricities are not defined in a useful way in Directed Acyclic Graphs (DAGs) using the standard measure of distance, since for any two nodes, there is no path between them in one of the two directions. So it is natural to consider the distance between two nodes as the length of the shortest path in the direction in which this path exists, motivating the definition of the min-distance. The min-distance between two nodes $u$ and $v$ is the minimum of the shortest path distances from $u$ to $v$ and from $v$ to $u$. 

As with the standard distance problems, the Strong Exponential Time Hypothesis [Impagliazzo-Paturi-Zane 2001, Calabro-Impagliazzo-Paturi 2009] leaves little hope for computing min-distance problems faster than computing All Pairs Shortest Paths, which can be solved in $\tilde{O}(mn)$ time. So it is natural to resort to approximation algorithms in $\tilde{O}(mn^{1-\epsilon})$ time for some positive $\epsilon$. 
Abboud, Vassilevska W., and Wang [SODA 2016] first studied min-distance problems achieving constant factor approximation algorithms on DAGs, and 
Dalirrooyfard \textit{et al} [ICALP 2019] gave the first constant factor approximation algorithms on general graphs for min-diameter, min-radius and min-eccentricities. Abboud \textit{et al} obtained a $3$-approximation algorithm for min-radius on DAGs which works in $\tilde{O}(m\sqrt{n})$ time, and showed that any $(2-\delta)$-approximation requires $n^{2-o(1)}$ time for any $\delta>0$, under the Hitting Set Conjecture. 
We close the gap, obtaining a $2$-approximation algorithm which runs in $\tilde{O}(m\sqrt{n})$ time. As the lower bound of Abboud \textit{et al} only works for sparse DAGs, we further show that our algorithm is conditionally tight for dense DAGs using a reduction from Boolean matrix multiplication. 
Moreover, Abboud \textit{et al} obtained a linear time $2$-approximation algorithm for min-diameter along with a lower bound stating that any $(3/2-\delta)$-approximation algorithm for sparse DAGs requires $n^{2-o(1)}$ time under SETH. We close this gap for dense DAGs up to an additive factor, by obtaining an $O(n^{2.350})$-time near-$3/2$-approximation algorithm, i.e. an algorithm which achieves a multiplicative approximation factor of $3/2$ plus an additive error, and showing that the approximation factor is unlikely to be improved within $O(n^{\omega - o(1)})$ time under the high dimensional Orthogonal Vectors Conjecture, where $\omega$ is the matrix multiplication exponent. 
\end{abstract}
\vfill
\pagebreak
\section{Introduction}

Among the most fundamental graph parameters that have been extensively studied are the diameter, radius and eccentricities \cite{chung, Hakimi, ChepoiD94,eppstein-planar-jv, aingworth, Corneil01, Chepoi02, Dvir04, BenMoshe, BeKa07, WN08, Yuster10, Chan12, FHW12, WeYu13, RV13, ChechikLRSTW14, AbboudGW15, BCH+15} (and many others). The eccentricity of a vertex $v$ is the largest distance between $v$ and any other vertex. The diameter is the maximum eccentricity of a vertex in the graph, thus the distance between the two farthest nodes, and the radius is the minimum eccentricity, measuring the maximum distance to the most central node. 

All of these parameters depend on the definition of the distance between two nodes. In undirected graphs, the distance between two vertices is just the shortest path distance $d(\cdot,\cdot)$ between them, which is symmetric. However, in directed graphs, this standard measure of distance $d$ is not necessarily symmetric, since for two nodes, $d(u,v)$ may not equal $d(v,u)$.

Several notions of a ``symmetric'' distance for directed graphs have been studied. Cowen and Wagner \cite{CW99} define the \textit{roundtrip distance}, 
which for two vertices $u$ and $v$ is just $d(u,v)+d(v,u)$. Abboud, Vassilevska W., and Wang \cite{avw} define the \textit{max-distance}, which is $\max\{d(u,v),d(v,u)\}$, and the \textit{min-distance}, which is $\min\{d(u,v),d(v,u)\}$.

%Given a directed graph $G = (V, E)$ and vertices $u, v \in V$, the $\textit{min-distance}$ $\dm(v, w)$ is defined as $\min(d(v, w), d(w, v))$. 
Each of these notions of distance has a particular application \cite{d19}. In this paper, we focus on the min-distance $\dm(\cdot,\cdot)$. 
The min-distance characterizes a quantity of real-world relevance: for instance, a patient may visit a doctor or a doctor may visit a patient, and if they are in a hurry the min-distance between them may matter. %Min-distance is a natural way to define a \textit{symmetric} notion of distance in directed graphs; in contrast the standard definition of distance in directed graphs is not symmetric.
Min-distance is a particularly natural notion of distance in directed acyclic graphs (DAGs), where the standard notion of distance is infinite in at least one direction for any given pair of vertices in a DAG. %\mina{added the following sentense} \jenny{edited a little, seems good}
For example, in a topologically ordered DAG where the edges are directed from left to right, the min-diameter is simply the largest distance $d(u, v)$ where $u$ is to the left of $v$.

More formally, for a vertex $v \in V$, the \textit{min-eccentricity} $\epsilon(v)$ is  $\max_{w\in V} \dm(v, w)$, or in other words, the largest min-distance between $v$ and any other vertex. The \textit{min-diameter} of a graph is $\max_{v \in V} \epsilon(v)$. Note that the min-diameter is the only meaningful notion of diameter for DAGs: all other notions are infinite. The \textit{min-radius} of a graph is $\min_{v \in V} \epsilon(v)$. A \textit{center} is a vertex whose min-eccentricity is equal to the min-radius of the graph.

All-Pairs Shortest Paths (APSP) is the problem of computing the distance between $u$ and $v$ for every pair of vertices $u, v \in V$.
In a graph $G$ with $m$ edges, $n$ vertices, and nonnegative edge weights polynomial in $n$, APSP can easily be computed in  $\tilde{O}(mn)$ time\footnote{The tilde hides polylogarithmic factors.}, by running Dijkstra's algorithm from every vertex\footnote{Faster algorithms are known by Pettie \cite{Pettie02} and Pettie and Ramachandran \cite{pettie2002computing} for sparse graphs.}. Computing eccentricities, diameter, or radius with any of the notions of distance is no harder than computing APSP. 

For the standard notion of distance, under the Strong Exponential Time Hypothesis (SETH) \cite{ipz1,seth2}, there is no \textit{truly subquadratic} time algorithm for diameter (and thus nor for eccentricities) in unweighted graphs: that is, no such algorithm runs in time $O(m^{2-\epsilon})$ for $\epsilon > 0$ \cite{RV13}. This lower bound also holds for the other notions of diameter (and eccentricities)  \cite{d19}.
For radius, the same lower bound holds but under the Hitting Set Conjecture \cite{avw}. 

Since quadratic time is expensive on large graphs, we resort to approximation algorithms. Many constant factor approximation algorithms were known for all notions of diameter, eccentricities and radius, except for the min-distance notion until recently. For example, for the standard diameter and roundtrip diameter there is a folklore linear time $2$-approximation algorithm, and for max-diameter and standard diameter, a conditionally tight $3/2$-approximation algorithm is known in $\tilde{O}(m\sqrt{n})$ time \cite{RV13}. 

Only recently Dalirrooyfard \textit{et al} \cite{d19} showed constant factor approximation algorithms for min-distance problems in general graphs that run in $O(mn^{1-\epsilon})$ time for some fixed $\epsilon>0$. More specifically, they obtained a $3$-approximation algorithm for min-diameter in $\tilde{O}(m\sqrt{n})$ time, a $(3+\delta)$-approximation algorithm for min-radius in $\tilde{O}(m \sqrt n/\delta)$ time, and a $(3+\delta)$-approximation algorithm for min-eccentricities in $\tilde{O}(m \sqrt n/\delta^2)$ time, for any $\delta>0$. 

The reason it is hard to obtain approximation algorithms for min-diameter, min-radius, and min-eccentricities is that min-distance does not obey the triangle inequality. Hence the typical approaches to find algorithms that work for other notions of distance do not work for min-distance, as they crucially rely on the triangle inequality.

On the bright side, since DAGs have more structure, it is easier to find algorithms for them. The best known subquadratic time algorithm for min-diameter in DAGs is a linear time $2$-approximation  algorithm, and the best subquadratic time algorithm for min-radius is a $3$-approximation algorithm in $\tilde{O}(m\sqrt{n})$ time \cite{avw}. However, neither of these algorithms were proven to be conditionally tight. 

%\jenny{I added a sentence here to make the transition smoother. Is it true that these are the only previously known conditional lower bounds?}mina:yup 
Previously, the only known conditional lower bounds for these problems were due to Abboud, Vassilevska W., and Wang \cite{avw}. They showed that under the Orthogonal Vectors Conjecture from fine-grained complexity (and consequently under SETH \cite{w05}), there is no $(3/2-\delta)$-approximation algorithm for any $\delta>0$ for min-diameter which runs in truly subquadratic time on sparse DAGs. Moreover, under the Hitting Set Conjecture, there is no $(2-\delta)$-approximation algorithm for any $\delta > 0$ for min-radius which runs in truly subquadratic time on sparse DAGs.

%There are known barriers to computing, or even closely approximating, min-diameter, min-radius, or min-eccentricities in \textit{subquadratic} time, meaning $O(m^{2-\delta})$ time for constant $\delta > 0$. For min-diameter, the barrier is the Strong Exponential Time Hypothesis, or more specifically the Orthogonal Vectors Conjecture which was shown in \cite{w05} to be implied by SETH. For min-radius, the barrier is the Hitting Set Conjecture. Indeed, in \cite{avw}, Abboud, Vassilevska W., and Wang showed the following conditional lower bounds:

%Proposition \ref{LB-rad} implies a conditional lower bound for approximating min-eccentricities as well, since for $\delta > 0$ a subquadratic $(2-\delta)$-approximation algorithm for min-eccentricities implies a subquadratic $(2-\delta)$-approximation algorithm for min-radius.

%Several subquadratic approximation algorithms for min-distance problems are known. In \cite{avw}, Abboud, Vassilevska W., and Wang found an $O(m\log n)$-time 2-approximation algorithm for min-diameter in DAGs and an $O(m\sqrt n \log n)$-time 3-approximation algorithm for min-radius in DAGs. Subquadratic approximation algorithms for these problems were obtained for general directed graphs in \cite{d19} by Dalirrooyfard \textit{et al}, who found an $\tilde{O}(m \sqrt n)$-time $3$-approximation algorithm for min-diameter, an $\tilde{O}(m \sqrt n/\delta)$-time $(3+\delta)$-approximation algorithm for min-radius, and an $\tilde{O}(m \sqrt n/\delta^2)$-time $(3+\delta)$-approximation algorithm for min-eccentricities.

\subsection{Our results}
We obtain fast algorithms for min-diameter, min-eccentricities and min-radius with improved approximation factors. Our results can be seen in Table \ref{tab:results}.

\begin{table}[]
    \centering
    \begin{tabular}{|c|c|c|c|}
    \hline
    Problem & Upper bound & Lower bound & Reference  \\ \hline
    min-diameter &  $2$ in $O(m)$  & $\mathbf{(\frac{3}{2}-\delta)}$ needs $m^{2-o(1)}$ & \cite{avw} \\ \cline{2-4}
     & $\mathbf{\left(\frac{3}{2}, 1\right)}$ in $O(n^{2.350})$ (dense, unweighted)  & $\mathbf{(\frac{3}{2}-\delta)}$ needs $n^{\omega-o(1)} \ (\ast)$ & this work \\ \hline
      min-radius &  $3$ in $\tilde{O}(m\sqrt{n})$  & $\mathbf{(2-\delta)}$ needs $m^{2-o(1)}$ & \cite{avw}  \\ \cline{2-4}
     & $\mathbf{2}$ in $\tilde{O}(\min(m\sqrt{n},m^{2/3}n))$  & $\mathbf{(2-\delta)}$ needs $n^{\omega-o(1)} \ (\ast) $ & this work  \\ \cline{2-4}
     & $k$ in $\tilde{O}(\min(mn^{1/k}, m^{\frac{2^{k-1}}{2^k -1}}n))$  &  & this work  \\\hline
     min-eccentri. &$3+\delta$ in $\tilde{O}(m\sqrt{n}/\delta^2)$  &  & \cite{d19} \\ \cline{2-4}
     &$k+\delta$ in $\tilde{O}(\min(mn^{1/k}/\delta, m^{\frac{2^{k-1}}{2^k -1}}n/\delta))$  &  & this work \\ \hline
     
    \end{tabular}
    \caption{Results on min-distance problems on DAGs.  The $(\ast)$ marks lower bounds that are for dense DAGs. Our $(2-\delta)$ lower bound for min-radius is based on Triangle Detection and our $(\frac{3}{2}-\delta)$ lower bound for min-diameter is based on high dimensional OV.  Our $k$ and $(k+\delta)$-approximation algorithms are for any integer $k\ge 2$. Conditionally tight bounds are in bold. We note that a  $(c,a)$-approximation for a quantity $D$ is a quantity $D'$ such that $D\leq D'\leq cD+a$. %\jenny{the log Ms should actually be in all of our min-radius and min-ecc runtimes. I think this will take too much space, though, perhaps we should omit it and state in the caption that this is for edge weights polynomial in n or something? also, our 2-approx is actually $\tilde{O}(\min( m\sqrt{n}, m^{2/3}n))$ which is smaller than $n^\omega$, so perhaps we can get rid of the asterisk?}  % \mina{does the min-rad alg of \cite{avw} work for min-ecc?} \jenny{no but \cite{d19} does, (3+delta) with $\delta^2$ in the denom and (5+delta) with $\delta$ in the denom. \\ also, what is the reason why comb. is only going in the LB side? \\ and another thing, both our lower bounds are based on BMM, right? so what is the reason to present one as $n^\omega$ and one as $n^3$ comb?}\mina{\cite{d19} result is not for DAGs though. the reason is that if we don't consider combinatorial, then it is based on triangle detection not BMM. which we can state anyway.} \jenny{\cite{d19} is for general directed graphs of which DAGs are a subset so their algorithms still work in DAGs. }
    %\mina{changed the table.}\jenny{nice!}
    } 
    \label{tab:results}
\end{table}

\subsubsection*{Min-Eccentricities and Min-Radius}

We obtain the first known subquadratic time $(2+\delta)$-approximation algorithm for min-eccentricities in DAGs for any $\delta > 0$, and the first known subquadratic time 2-approximation algorithm for min-radius in DAGs. These algorithms run in time $\tilde{O}(\min(m\sqrt{n}/\delta, m^{2/3}n)/\delta)$ and $\tilde{O}(\min(m\sqrt{n}, m^{2/3}n))$ respectively. %\jenny{not quite, the min-ecc algorithm runtime is divided by delta...also, should we say ``$O(\min(f(m, n), g(m,n)))$'' or is it better to say something like ``simultaneously $O(f(m,n))$ and $O(g(m,n))$''?} 
Note that our algorithms in this section are combinatorial: they do not exploit fast matrix multiplication and are potentially practical. %\jenny{i think this informal definition of combinatorial can replace our footnote} %\footnote{The notion of a “combinatorial” algorithm does not have a formal definition. Intuitively, such algorithms are not only theoretically but also practically efficient. One can think of combinatorial algorithms as having low leading constants in their running times.}. %\jenny{The one in Theorem \ref{unweighted} isn't.}mina: added for this section 
Our results are \textit{conditionally optimal} in both sparse and dense graphs: For sparse graphs,
if the Hitting Set Conjecture is true, then our min-radius result is tight and our min-eccentricity result is essentially tight, in the sense that no approximation factor smaller than 2 can be achieved in subquadratic time for either of these problems \cite{avw}. For dense graphs, our $2$-approximation algorithm works in $\tilde{O}(n^{7/3})$ time, and we show that there is no ($2-\delta$)-approximation algorithm for min-radius (and hence min-eccentricities) in $O(n^{\omega-\epsilon})$ for $\epsilon > 0$, if the best algorithm for Triangle Detection runs in time $\Omega(n^{\omega-o(1)})$. %\jenny{Triangle Detection Conjecture? currently this seems like a type mismatch since TD is a problem not a hypothesis}
Here $\omega<2.37286$ \cite{matrixmult2020} is the exponent of matrix multiplication. %Moreover, we show that there is no ($2-\delta$)-approximation combinatorial algorithm for min-radius (and hence min-eccentricities) in $n^{3-o(1)}$ time if there is no  truly subcubic combinatorial algorithm for Boolean matrix multiplication. \jenny{have not dealt with virginia's comment here}

%we reduce Triangle Detection to ($2-\delta$)-approximation of min-radius, and hence  achieving a lower bound of $n^{3-o(1)}$ lower bound for both min-radius and min-eccentricities.

More generally, we obtain a series of algorithms trading off runtime and accuracy.% hierarchy of approximation algorithms for these problems, parametrized by a positive integer $k$ which governs the tradeoff between runtime and approximation accuracy:

\begin{theorem}\label{main-thm}
For  integer $k \geq 2$ and every $\delta > 0$, there is a $(k+\delta)$-approximation algorithm for min-eccentricities in DAGs which runs in $\tilde{O}(\min(mn^{1/k}/\delta,  m^{2^{k-1}/(2^k -1)}n/\delta))$ time.\\
\indent For every integer $k \geq 2$, there is a $k$-approximation algorithm for min-radius in DAGs which runs in $\tilde{O}(\min(mn^{1/k}, m^{2^{k-1}/(2^k -1)}n))$ time.
\end{theorem}

As mentioned earlier, the case $k=2$ gives a 2-approximation algorithm for min-radius running in time $\tilde{O}(\min(m\sqrt{n}, m^{2/3}n))$. For $m = \tilde{O}(n^{1.5})$, this matches the runtime and improves the approximation factor of the previous best known algorithm for this problem (from \cite{avw}). %\mina{the have the same bound as us for all m right?} \jenny{when m is small this is $\tilde{O}(m\sqrt{n})$ which is the same as \cite{avw}, when m is larger our runtime is better than \cite{avw}}.
For $m = \omega(n^{1.5 + o(1)})$%\footnote{We say that a function $g(n) $ is $\omega{(f(n))}$ for some function $f$, if for any constant $c>0$, $g(n)>cf(n)$ for all $n>n_0$ for some $n_0>0$.}
, it improves both the approximation factor and the runtime. 

%\jenny{here is a paragraph about techniques}
Our min-eccentricity $(2+\delta)$-approximation algorithm borrows a key idea from the 3-approximation algorithm of \cite{avw} and combines it with a new binary search technique. The idea is to partition the DAG into intervals and do local APSP searches to find local paths, then combine these local paths with ``outer'' paths to guarantee a low enough min-distance to any vertex in the graph. In \cite{avw}, these outer paths were found by using a clever choice of intervals; our algorithm instead applies binary search to find sets which can be used as jumping-off points for the outer paths, allowing us to shorten the lengths of these paths and also allowing us to approximate all min-eccentricities, not only min-radius. Our  $(k+\delta)$-approximation algorithm is achieved by recursively running our approximation algorithm on the intervals instead of running local APSP, which allows us to improve the runtime.
%\mina{a paragraph about techniques? what is new here that was not in \cite{avw}?} \jenny{mostly just a new application of binary search to finding sets that can serve as a halfway point in the 2-approximation paths, or more generally a 1/kth-way point in the k-approximation. the k-approximation is a recursion. I'll write this up tomorrow (Tuesday)}

For sparse graphs, Abboud, Vassilevska W., and Wang \cite{avw} already showed that a $(2-\delta)$-approximation for min-radius needs $\Omega(m^{2-o(1)})$ time under the Hitting Set Conjecture, so our $2$-approximation algorithm is conditionally tight for sparse graphs. We show that the approximation factor of our algorithm is conditionally tight for the dense case as well by reducing Triangle Detection to $(2-\delta)$-approximation of min-radius for any $\delta>0$. The best running time for Triangle Detection in $n$-node graphs is conjectured to be $\Omega(n^{\omega - o(1)})$ by many papers (see for example \cite{abboud2018if,bringmann2018clique}), where $\omega<2.37286$ \cite{matrixmult2020} is the exponent of fast matrix multiplication. Note that, since $m = O(n^2)$, our algorithm runs in $\tilde{O}(n^{7/3})$ time, which is faster than $O(n^\omega)$ for the current best bound on $\omega$. Since the algorithm of Theorem \ref{main-thm} is combinatorial, if we restrict to combinatorial algorithms  then there is no \textit{truly subcubic} (meaning $O(n^{3-\epsilon})$ for $\epsilon > 0$) time $(2-\delta)$-approximation algorithm for min-radius provided that there is no truly subcubic time combinatorial algorithm for Boolean matrix multiplication (BMM). This is because BMM and Triangle Detection are subcubic equivalent \cite{williams2018subcubic}. Note that our reduction graph in Theorem \ref{thm:minradlb} is an unweighted DAG.

%\jenny{should we define subcubic equivalent somewhere?}mina:done

\begin{theorem}
\label{thm:minradlb}
If there is a $T(n,m)$-time algorithm for $(2-\delta)$-approximation of min-radius in $O(n)$-node $\tilde{O}(m)$-edge DAGs for some $\delta>0$, then there is an $\tilde{O}(T(n,m) + m)$-time algorithm for Triangle Detection on graphs with $n$ nodes and $m$ edges.
\end{theorem}

\begin{corollary}
Assuming the best algorithm for Triangle Detection runs in time $\Omega(n^{\omega-o(1)})$, there is no algorithm for $(2-\delta)$-approximation of min-radius in $n$-node dense DAGs  that runs in time $O(n^{\omega-\epsilon})$ for any $\delta,\epsilon > 0$. %\mina{$n^{\omega}$ or $n^{\omega-o(1)}$}

Moreover, there is no $O(n^{3-\epsilon})$-time combinatorial algorithm for $(2-\delta)$-approximation of min-radius in $n$-node dense DAGs with $\epsilon, \delta > 0$ if there is no $O(n^{3-\epsilon'})$-time combinatorial algorithm for BMM with $\epsilon' > 0$.

\end{corollary}

%\mina{the next result is a bit too complicated for the intro. Maybe just say that we improve the prevoius result (what is the previous result?) and leave it very informal here.} \jenny{here is a less complicated version, do you think this is still too long/complex?}mina: looks good!
\paragraph*{Improving the running time using Fast Matrix Multiplication} In DAGs with small integer edge weights, we further improve the running times for all $k$ in Theorem \ref{main-thm} by applying a result of Zwick in \cite{zwickbridge} on the runtime of APSP in such graphs. We describe our result in more detail in Section 2. %(using fast matrix multiplication). 
In particular, in DAGs with constant integer edge weights, including unweighted DAGs, our result in the case $k = 2$ is as follows:

\begin{theorem}\label{unweighted}
For every $\delta > 0$, there is an $\tilde{O}(\min(m\sqrt{n}/\delta, m^{0.605}n/\delta))$-time $(2+\delta)$-approximation algorithm for min-eccentricities in DAGs with constant integer edge weights. \\
\indent There is an $\tilde{O}(\min(m\sqrt{n}, m^{0.605}n))$-time $2$-approximation algorithm for min-radius in DAGs with constant integer edge weights.
\end{theorem}

\subsubsection*{Min-Diameter}

%\mina{to self: min-diam is harder than min-radius (for some reason?), and so to obtain fast algorthms we use sparse fast matrix multiplication. However the running time that we obtain is better than what high dimensional OV says.}
 We obtain a $(3/2, 1/2)$-approximation algorithm for min-diameter in unweighted DAGs, where the multiplicative approximation factor is conditionally optimal in dense graphs. Specifically, our algorithm improves on the standard APSP runtime for any graph with $m = \omega(n^{1+o(1)})$ edges. This is the first known near-$3/2$-approximation algorithm for min-diameter in dense DAGs that runs faster than the best constant factor approximation algorithm for APSP, which runs in $\tilde{O}(n^\omega)$ time in unweighted directed graphs \cite{zwickbridge}. 
%with constant diameter \cite{alon92}, and can be $(1+\epsilon)$-approximated in $\tilde{O}(n^\omega/\epsilon)$ time in unweighted graphs \cite{zwickbridge}. Our algorithms are faster than these running times. \jenny{to fix}

\begin{theorem}\label{mindiam} %\mina{can we write the bound in a more general form, and then discuss more specific running times afterwards?} \jenny{it'll probably be really gross but I can try} mina: that's fine then
There is an $O(m^{0.414}n^{1.522} + n^{2+o(1)})$-time $(3/2, 1/2)$-approximation algorithm for min-diameter in unweighted DAGs.
\end{theorem}

This algorithm relies on the sparse matrix multiplication algorithm of Yuster and Zwick \cite{yzmm}. In dense graphs with $m = O(n^2)$, its runtime is $O(n^{2.350})$. In relatively sparse graphs, with $m = O(n^{1.154 + o(1)})$, the second term dominates, so the runtime is $O(n^{2+o(1)})$.

%\mina{added the following paragraph for our techniques.} \jenny{edited a little, I changed dijsktra to BFS since it's unweighted, and also I use L(v) and R(v) to refer to other things later so I renamed those}
Our techniques, which mix known diameter techniques with sparse matrix multiplication, are informally as follows: We first construct a covering set, which will intersect any sufficiently large set. We run BFS from all vertices in the covering set, and check whether any min-distances found were large. If not, then for each vertex $u$, we will define a set of vertices that are relatively ``close'' to $u$ on its right; if this set is large it will intersect the covering set, allowing us to find paths from $u$ to some vertices to its right, using a ``close'' vertex in the covering set as a jumping-off point. The remaining vertices $w$, for which this method did not construct a $u \to w$ path, must have the property that any $u \to w$ path must intersect a relatively small subset of the set of vertices ``close'' to $u$ (note that this set may have been small to begin with, in which case we can skip the previous step). Symmetrically, for each vertex $w$ we can construct the corresponding relatively small subset of vertices ``close'' to $w$ on its left, and then to bound the min-distance between $u$ and $w$ we check whether these two small subsets share a vertex in common. We use sparse matrix multiplication to detect this set intersection.

%The conditional lb says this even though the construction is a sparse instance. there is no algorithm for 3/2 if better than that time, neither for dense nor for sparse graphs. We provide a 3/2 approx alg for dense graphs (can't be applied on the lower bound of [3] to get a faster algortihm for the reduction graph). This algorithm is optimal because ov high dimensional ov.   

%\mina{I added the following paragraph}
The conditional lower bound of \cite{avw} says that if the Orthogonal Vectors Conjecture is true then min-diameter cannot be $(3/2 - \delta)$-approximated in truly subquadratic time in sparse graphs. There is no known $3/2$-approximation algorithm for min-diameter on DAGs that works faster than APSP, neither for dense graphs nor for sparse graphs. %\jenny{hmm I'm not understanding the previous sentence, \cite{avw} doesn't say anything about dense graphs. it just says that 3/2 - delta is impossible in sparse DAGs.}\mina{that sentence doesn't refer to the lower bound. It says that there is no 3/2 algorithm in any case. That's what i meant.} 
%\jenny{as in ``there is no algorithm known'' or as in ``there is no algorithm period''?} \mina{yeah i meant known, fixed}
So the question is: Is $3/2$ the right multiplicative bound for inapproximability of min-diameter in DAGs? We answer this question in the affirmative for dense DAGs. Theorem \ref{mindiam} gives the first near-$3/2$-approximation algorithm that works faster than APSP, and it is optimal (up to an additive factor) conditioned on \textit{high dimensional OV} using the same reduction as \cite{avw}. High dimensional OV can be used for obtaining lower bounds for dense graphs. In high dimensional OV, the dimension of the vectors can be as big as $O(n)$, and using a simple reduction to Boolean matrix multiplication, %\jenny{I follow the intuition here but if this is written out in some paper somewhere I think we should cite it}\mina{I didn't find any references for it, explained it more in prelim}
the best known algorithm for it is in time $O(n^{\omega})$. %\mina{have we mentioned it anywehre else?} \jenny{yes but this is the first time it's mentioned; I'm not sure whether the bound/citation for bound should also go in 1.2 (right now it's there too)}. mina: good i meant first time. jenny: update, this is no longer the first place it's mentioned

High dimensional OV gives a conditional lower bound of $\Omega(n^{\omega-o(1)})$ time for $(3/2-\delta)$-approximation of min-diameter for any $\delta>0$. %\jenny{I think this should be given the same amount of discussion than HS and regular OV are, and added to the preliminaries and such} mina:done
Our algorithm gives an upper bound of $O(n^{2.350})$  for $m=\Theta{(n^2)}$, which is faster than $O(n^\omega)$ for the current best bound on $\omega$. We note while we provide conditionally tight results for the multiplicative approximation factor achievably in dense DAGs, the gap between the lower bound and upper bound for computing min-diameter on sparse DAGs is still open. %\mina{update} \jenny{changed it to the decimal}

\subsection{Preliminaries}
%\mina{define OV conjecture, hitting set conjecture, triangle detection.} \jenny{copy pasting HS below since I already had it stored somewhere, will do the others tomorrow}

All graphs in this paper are  directed graphs. Given a graph $G$, $n$ denotes the number of vertices and $m$ denotes the number of edges. We will assume $m \geq n-1$ since otherwise all min-eccentricities are infinite, a case that is easily checked. %We use $M$ to denote the maximum weight of any edge in $G$. 
All edge weights are assumed to be nonnegative and polynomial in $n$; if $w_{max}$ is the maximum edge weight and $w_{min}$ is the minimum edge weight, we let $M=\max\{w_{max},1/w_{min}\}$. %cut for length, and note that in case of integer weights, $M=W_{max}$. %\jenny{is this a reasonable way to do the min weight assumption?}\mina{maybe let's just say that if $W_{max}$ is the maximum and $W_{min}$ is the minimum, $M=max\{W_{max},1/W_{min}\}$, and note that in case of integers $M=W_{max}$.}
We write $G[S]$ to denote the subgraph of $G$ induced by vertex set $S$. For a vertex $v$, we write $N_D^{\text{in}}(v)$ (respectively, $N_D^{\text{out}}(v)$) to denote the set of vertices $u$ such that $d(u, v) \leq D$ (respectively, $d(v, u) \leq D$).  

For $v \in V$ and $W \subseteq V$, we define $\dm(W, v) = \dm(v, W)$ as $\min_{w \in W} \dm(v, w)$, and we define the min-eccentricity of $W$ as $\epsilon(W) = \max_{v \in V} \dm(W, v)$.

Given two sets $U, W \subseteq V$, if every $u \in U$ appears prior to (respectively, after) every $w \in W$ in a topological ordering of the vertices of $G$, we say that $U$ is the left (respectively, right) of $W$ with respect to the topological ordering. When $U$ or $W$ consists of a single vertex $\{x\}$, we omit the brackets. If $W \subseteq U \subseteq V$, we denote the subset of vertices in $U$ that lie to the left (right) of $W$ by $L_U(W)$ (respectively, $R_U(W)$). If $U = V$, we omit the subscript. A vertex set $W$ is called \textit{topologically consecutive} with respect to a topological ordering if its vertices are consecutive; i.e., if $W = V \setminus (L(W) \cup R(W)$). In general, the relevant topological ordering will be clear, and we will omit reference to it. 
%The \textit{min-diameter} of $G$ is $\max_{v, w \in V} \dm(v,w)$.

Let $\omega(1, r, 1)$ be the exponent of the runtime of multiplying $n \times n^{r}$ by $n^{r} \times n$ matrices. Let $\omega = \omega(1, 1, 1)$ be the square matrix multiplication exponent. \cite{matrixmult2020} showed that $\omega < 2.37286$.

%By a truly subcubic (subquadratic) algorithm for an $n$-node graph we mean an algorithm that runs in time $O(n^{3-\epsilon})$ ($O(n^{2-\epsilon})$) for some constant $\epsilon>0$. jenny: removed this because we use both these terms in the intro and i've edited so that they're defined when used

For specifying lower bounds, we use the following problems with their corresponding running time conjectures. 
\paragraph*{Orthogonal Vectors (OV)} Given two lists $A, B$ of $n$ $d$-dimensional Boolean vectors, determine whether there are vectors $a\in A$ and $b\in B$ such that $a$ and $b$ are \textit{orthogonal}; i.e. there is no $i\in [d]$ such that the $i$th bits of both $a$ and $b$ are $1$. When $d=\Omega(\log{n})$, the \textit{OV Conjecture} \cite{w05} says that there is no algorithm that can solve the OV problem in time $O(n^{2-\epsilon})$ for any fixed $\epsilon>0$. The OV Conjecture is implied by the Strong Exponential Time Hypothesis (SETH) \cite{w05}.

\paragraph*{High Dimensional Orthogonal Vectors}
In high dimensional OV, the dimension $d$ can be as high as $O(n)$. There is a simple reduction from high dimensional OV to matrix multiplication: Given two lists $A=\{a_1,\ldots,a_n\},B=\{b_1,\ldots,b_n\}$ of $d$-dimensional Boolean vectors, let $M$ and $N$ be two $n \times d$ and $d \times n$ Boolean matrices, where $M[i,j]=1$ if $a_i$ is $1$ in bit $j$, and $N[j,k]=1$ if $b_k$ is $1$ in bit $j$, for $j=1,\ldots,d$ and $i,k=1,\ldots,n$. If $MN$ has a zero entry, the vector pair corresponding to that entry are orthogonal. This gives a $O(n^\omega)$ algorithm for high dimensional OV, and there are no faster algorithms known for it up to polylogarithmic factors. Moreover, OV is equivalent to the problem of distinguishing diameter $2$ vs $3$ \cite{RV13}, and so high dimensional OV is equivalent to distinguishing diameter $2$ vs $3$ in dense graphs. A well-known open problem is whether diameter $2$ vs $3$ can be solved faster than matrix multiplication (see for example \cite{aingworth}). %ask whether this problem can be solved faster than matrix multiplication\footnote{More precisely, Aingworth et al \cite{aingworth} ask if there is a combinatorial algorithm running in time $O(n^{3-o(1)})$ for distinguishing between graphs of diameter $2$ and $3$.}. 
Hence, it is conjectured that high dimensional OV cannot be solved in $O(n^{\omega-\epsilon})$ time for any $\epsilon>0$.

\paragraph*{Hitting Set (HS)} Given two lists $A, B\in\{0, 1\}^d$, determine whether there is a vector $a\in A$ that is not orthogonal to any vector $b\in B$. When $d=\Omega(\log{n})$, the \textit{Hitting Set Conjecture} \cite{avw} says that there is no algorithm that can solve the Hitting Set problem in time $O(n^{2-\delta})$ for any fixed $\delta>0$.

\paragraph*{Boolean Matrix Multiplication (BMM)} We abbreviate multiplying two Boolean $n\times n$ matrices over the (AND, OR)-semiring by BMM. It is conjectured that there is no combinatorial algorithm solving BMM in $O(n^{3-\epsilon})$ time for any fixed $\epsilon>0$, and the best algebraic algorithm for it is in $O(n^{\omega + o(1)})$ time for $\omega<2.37286$ \cite{matrixmult2020}. 
\paragraph*{Triangle Detection \cite{williams2018subcubic}} Given a tripartite graph $G(A,B,C,E)$ where $A$, $B$ and $C$ are the three parts of the vertex set and $E$ is the edge set, determine if there are $a\in A$, $b\in B$, and $c\in C$ such that $abc$ is a triangle. Vassilevska W. and Williams \cite{williams2018subcubic} showed that considering only combinatorial algorithms, Triangle Detection and BMM are subcubic equivalent, meaning that a truly subcubic combinatorial algorithm in one results in a truly subcubic combinatorial algorithm in the other. Moreover, the best (algebraic) algorithm for Triangle Detection is through BMM. Thus the best running time for Triangle Detection is $O(n^{\omega})$, and it is conjectured (see for example \cite{abboud2018if,bringmann2018clique}) that there is no algorithm faster than $O(n^{\omega})$ for detecting a triangle.  

 %\jenny{hmm, maybe this should be a footnote that appears the first time we use the word combinatorial?}done

%\begin{conjecture}[\cite{avw}]{\label{HS}}

%There is no $\delta > 0$ such that for all $c \geq 1$, there is an algorithm that given two lists $\mathcal{A}, \mathcal{B}$ of $n$ subsets of a universe $U$ of size $|U| = c \log n$, can decide in time $O(n^{2-\delta})$ whether there exists a set in $\mathcal{A}$ that intersects (or ``hits'') every set in $\mathcal{B}$.
%\end{conjecture}

%Finally, we formally mention the lower bounds obtained by Abboud \textit{et al} obtained for approximating min-diameter and min-radius on DAGs.
%\begin{proposition}[\cite{avw}]{\label{LB-diam}}
%For $\delta > 0$, a $(3/2-\delta)$-approximation algorithm for min-diameter which runs in subquadratic time on sparse DAGs refutes the OV Conjecture and consequently also SETH.
%\end{proposition}

%\begin{proposition}[\cite{avw}]{\label{LB-rad}}
%For $\delta > 0$, a $(2-\delta)$-approximation algorithm for min-radius which runs in subquadratic time on sparse DAGs refutes the Hitting Set Conjecture.
%\end{proposition}

 %(Note that we do not use a subscript to specify min-eccentricity in this paper.)

\section{Min-Eccentricities and Min-Radius}

We present two different versions of our min-eccentricity and min-radius approximation algorithms, one which works in general weighted DAGs and is combinatorial and one with a lower runtime upper bound which only works in DAGs with small integer edge weights. The algorithms are identical except in how they compute APSP; the former computes APSP in the standard combinatorial way, while the latter uses Zwick's fast APSP algorithm for graphs with small integer edge weights. Here, $\mu(t)$ is the value satisfying $\omega(1, \mu(t), 1) = 1 + 2\mu(t) - t$.

\begin{theorem}[\cite{zwickbridge}]\label{zwick}
APSP can be computed in $O(n^{2+\mu(t)})$ time in directed graphs with integer edge weights bounded by $n^t$, where $t < 3 - \omega$. %note: this actually works with negative weights also
\end{theorem}

Both versions of our algorithms use a common technique to compute min-distances to and from a vertex set. Given a graph $G$ and a vertex set $W \subseteq V$, we construct a graph $G'$ by adding a vertex $y$ and adding weight-0 edges $(w, y)$ for all $w \in W$. We then run Dijkstra into $y$ in $G'$. We refer to this procedure as \textit{running Dijkstra into $W$}. The symmetric procedure, in which the weight-0 edges point out of an added vertex $y'$ and we run Dijkstra out of $y'$, will be referred to as \textit{running Dijkstra out of $W$}. Then for $x \in V$, $\dm(x, W) = \min(d(x, y), d(y', x))$, a value which we can now compute.  We added $|W|$ edges and ran Dijkstra in $G'$, so in total the procedure takes time $O(|W| + m \log n) = O(m \log n)$.

Our min-eccentricity and min-radius approximation algorithms will be based on the following proposition. Let $c_k(\tau) = \frac{2^{k-2}(1+\tau)}{2^{k-1}(1+\tau) -\tau}$.

\begin{proposition}\label{k-prop} For any $k \geq 2$, there is an $O(\min(mn^{1/k}\log^2 n, m^{2^{k-1}/(2^k -1)}n \log^2 n))$-time algorithm which takes as input a DAG $G$ and a parameter $r$, and certifies for each vertex $v$ that $\epsilon(v) > r$ or that $\epsilon(v) \leq kr$. \\
\indent In DAGs with integer edge weights bounded by $n^t$, where $t < 3 - \omega$, there is a version of this algorithm which runs in $O(\min(mn^{1/k}\log^2 n, m^{c_k(\mu(t))}n \log^2 n))$-time.
\end{proposition}

In \cite{lu18}, Le Gall and Urrutia showed that $\mu = \mu(0) < 0.529$.  Thus in DAGs with constant integer edge weights (so that $t = 0$), the runtime of the algorithm of Proposition \ref{k-prop} is $\tilde{O}(\min(mn^{1/k}, m^{c_k(0.529)}n))$ time. When $k=2$, $c_k(0.529) < 0.605$, leading to the special case stated in Theorem \ref{unweighted}.

The algorithms of Proposition \ref{k-prop} will be described and proven correct in subsection 2.1, and their runtimes will be analyzed in Lemma \ref{k-prop-runtime} in subsection \ref{min-ecc-appendix}. Then by binary searching over $r \in [0, Mn]$, these algorithms can be used to obtain the min-eccentricity approximation algorithms of Theorems \ref{k-approx-1} and \ref{k-approx-2} and the min-radius approximation algorithms of Theorems \ref{min-radius-1} and \ref{min-radius-2}.

%\mina{seems like the first part is the same as theorem 1. let's seperate the ''general'' theorem from the ''small integers'' theorem that uses fast mm.} \jenny{is the following adequately separated, or do you think there should be two separate theorems?}\mina{I was thinking more of seperate thms, cause one is combinatorial one uses FMM. The proofs can be the same though.}
\begin{theorem}\label{k-approx-1}
%\mina{remove log n since we have tilde} \jenny{are our edge weights polynomial in $n$ though?} \mina{doesn't matter, i have kept log(M)}
Let $k \geq 2$ be an integer. For any $\delta > 0$, there is an $\tilde{O}(\min(mn^{1/k}/\delta,$ $m^{2^{k-1}/(2^k -1)}n/\delta))$-time algorithm which, given a DAG $G$, outputs for every vertex $v \in V$ an estimate $\epsilon'(v)$ such that $\epsilon(v) \leq \epsilon'(v) < (k+\delta)\epsilon(v)$.
\end{theorem}

\begin{theorem}\label{k-approx-2}
 Let $k \geq 2$ be an integer. For any $\delta > 0$, there is an $\tilde{O}(\min(mn^{1/k}/\delta,$ $m^{c_k(\mu(t))}n/\delta))$ time algorithm which, given a DAG $G$ with integer edge weights bounded by $n^t$ for $t < 3 -\omega$, outputs for every vertex $v \in V$ an estimate $\epsilon'(v)$ such that $\epsilon(v) \leq \epsilon'(v) < (k+\delta)\epsilon(v)$.
\end{theorem}

%These two theorems can be proven in the same way. A proof is given in Appendix \ref{min-ecc-binary-search}.
\begin{proof}
First we have all the vertices as ``unmarked.'' We do binary search in $[0, Mn]$ by starting with $r=1$ in Proposition \ref{k-prop} and incrementing $r' = (1+\delta/k)r$ at each step. At each step, we run the algorithm given in Proposition \ref{k-prop}, and for each unmarked $v$ that is reported as having $\epsilon(v) \leq kr$, we set $\epsilon'(v) = kr$ and mark $v$. %We increment $r' = r(1+\delta/k)$.
At the end we set $\epsilon'(v) = \infty$ for any remaining unmarked vertices.

Suppose a vertex $v$ was marked at the step corresponding to $r$. Then $r/(1+\delta/k) < \epsilon(v) \leq kr$, so $\epsilon(v) \leq \epsilon'(v) = kr < (k+\delta)\epsilon(v)$. The binary search adds an $O(\log_{1+\delta/k} Mn) = O((\log Mn)/\delta)$ factor to the runtime. Since $\log Mn$ is polylogarithmic in $n$, this gives the time bounds stated.
\end{proof}

\begin{theorem}\label{min-radius-1}
Let $k \geq 2$ be an integer. There is an $\tilde{O}(\min(mn^{1/k},$ $m^{2^{k-1}/(2^k -1)}n))$-time algorithm which, given a DAG $G$, outputs an approximation $R'$ such that if $R$ is the min-radius of $G$, $R \leq R' < kR$. 
\end{theorem}

\begin{theorem}\label{min-radius-2}
Let $k \geq 2$ be an integer. There is an $\tilde{O}(\min(mn^{1/k},$ $m^{c_k(\mu(t))}n))$-time algorithm which, given a DAG $G$ with integer edge weights bounded by $n^t$ for $t < 3 -\omega$, outputs an approximation $R'$ such that if $R$ is the min-radius of $G$, $R \leq R' < kR$.
\end{theorem}

%These theorems are likewise proven the same way as each other. Their proof is given in Appendix \ref{min-rad-binary-search}.
\begin{proof}

We do binary search in $[0, Mn]$, running the algorithm given by Proposition \ref{k-prop} at each step as follows: %\mina{added the following} \jenny{looks good}
We keep two numbers $A_i$ and $B_i$ at step $i$ which are the lower bound and upper bound to the min-radius $R$. At step $1$ we have $A_1 = 0$ and $B_1 = Mn$. At step $i$, we have $A_i, B_i$ such that $A_i < R \leq B_i$. %(Here, $A_0 = 0$ and $B_0 = Mn$). 
Let $C_i = B_i - kA_i$. If $C_i$ is smaller than the minimum positive edge weight, then any path of length at most $B_i$ must have length at most $kA_i$, so in this case we terminate the binary search and let $R' = kA_i$. We now have $R \leq R' < kR$ as desired.

If $C_i$ is not smaller than the minimum positive edge weight, let $r = A_i + \frac{C_i}{k+1}$, and run the algorithm given by Proposition \ref{k-prop}. If the algorithm reports that there is a vertex $v$ with $\epsilon(v) < kr$, then let $A_{i+1} = A_i$ and $B_{i+1} = kr = kA_i + \frac{k}{k+1}C_i$, as we have the min-radius is between $A_{i+1}$ and $B_{i+1}$. %the min-radius is at least $A_{i+1} = A_i$ and is less than $B_{i+1} = kr = kA_i + \frac{k}{k+1}C_i$. 
Note that in this case $C_{i+1} = B_{i+1} - kA_{i+1} = \frac{k}{k+1}C_i$. 
Otherwise, if the algorithm reports that every vertex has $\epsilon(v) \geq r$, then the min-radius is at least $A_{i+1} := r = A_i + \frac{C_i}{k+1}$ and is less than $B_{i+1} := B_i$. In this case $C_{i+1} = B_i - k(A_i + \frac{C_i}{k+1}) = \frac{k}{k+1}C_i$. Thus, at each step, the size of $C_i$ shrinks by a factor of $\frac{k}{k+1}$. Hence, for constant $k$, the algorithm will in $O(\log Mn)$ steps find bounds $A_i, B_i$ such that $C_i$ is smaller than the minimum positive edge weight.

\end{proof}

\subsection{Algorithm Description and Correctness}
We now describe and prove the correctness of the algorithm of Proposition \ref{k-prop} by induction on $k$. For convenience, we use $k = 1$ as a base case; in this case we simply run an APSP computation. Our algorithm for $k > 1$ is as follows. %\mina{need to define $c_k$}. \jenny{it's defined earlier}

%\mina{maybe change t to something else cause we need a t up there that might be confused with this. } \jenny{changed to p}
 First, topologically sort the vertices and partition them into $p$ consecutive sets $W_1, \dots W_p$ of size $|W_i| = n/p$. The runtime-minimizing value of $p$ will be chosen later.

For each $i$, run Dijkstra to and from $W_i$. If $\epsilon(W_i) > r$, then we can report $\epsilon(w) > r$ for all $w \in W_i$. Otherwise, $\epsilon(W_i) \leq r$. In this case, we will apply Claim \ref{binarysearch}, below, twice. Recall that for $S \subseteq W \subseteq V$, $L_W(S)$ is the set of vertices in $W$ that are to the left of all vertices in $S$ in the topological ordering.

\begin{claim}\label{binarysearch}
Let $W \subseteq V$ be a topologically consecutive subset of a topologically ordered DAG $G$, and let $r$ be a parameter such that $\epsilon(W) \leq r$. In $O(m\log^2 n)$ time, one can find a nonempty topologically consecutive subset $S \subseteq W$ such that:

\begin{enumerate}[label=(\alph*)]
    \item $\epsilon(S) \leq r$.
    \item If $w \in L_W(S)$, $\epsilon(w) > r$.
    \item If $|S| > 1$, all vertices $s \in S$ satisfy $\epsilon(s) > r$.
\end{enumerate}

\end{claim}

\begin{proof}

We will use a binary search argument to find $S$. We will induct on an index $j$. Let $S^0 = W$. Assume that $S^j \subseteq W$ is topologically consecutive, that $\epsilon(S^j) \leq r$, and that for every $w \in L_W(S^j)$, $\epsilon(w) > r$. These all hold for $j = 0$. If $S^j = \{s\}$ consists of a single vertex, let $S = S^j$; then we are done.

Otherwise, let $S^j_L$ be the subset of $S^j$ containing its first $|S^j|/2$ vertices in the topological ordering and let $S^j_R = S^j \setminus S^j_L$. So $S^j_L$ and $S^j_R$ are the left and right halves of $S^j$, respectively; hence both $S^j_L$ and $S^j_R$ are topologically consecutive. See Figure \ref{fig:min-ecc-1}.

\begin{figure}[h]
    \centering
    \includegraphics[scale=0.38]{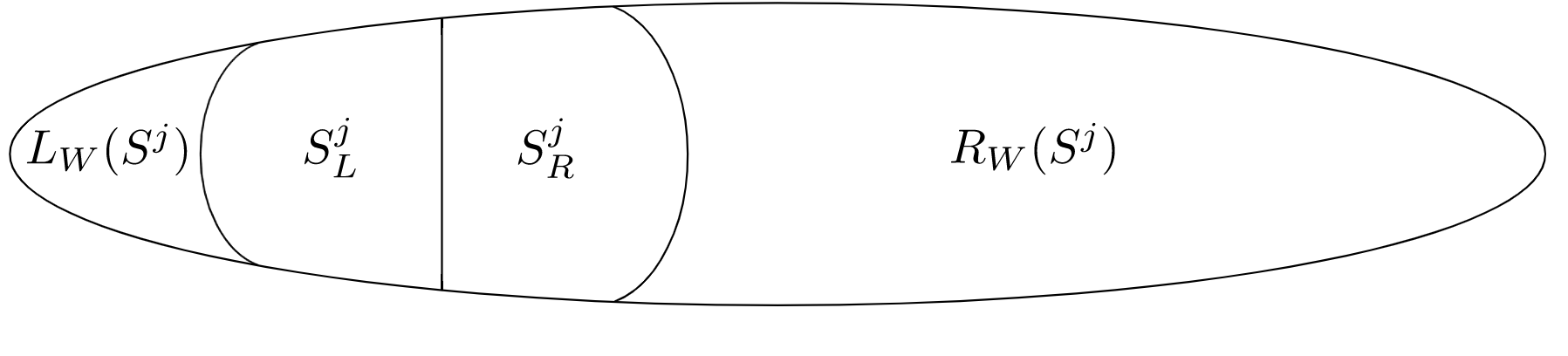}
    \caption{$S^j$ is partitioned into two halves, $S^j_L$ and $S^j_R$.}
    \label{fig:min-ecc-1}
\end{figure}

Run Dijkstra from $S^j_L$ and from $S^j_R$. If either of these sets has min-eccentricity at most $r$, we will continue the induction: If $\epsilon(S^j_L) \leq r$, we let $S^{j+1} = S^j_L$. Then $L_W(S^{j+1}) = L_W(S^j)$, so for every $w \in L_W(S^{j+1})$, $\epsilon(w) > r$. Alternatively, if $\epsilon(S^j_L) > r$ but $\epsilon(S^j_R) \leq r$, we let $S^{j+1} = S^j_R$. Then $L_W(S^{j+1}) = L_W(S^j) \cup S^j_L$, so for every $w \in L_W(S^{j+1})$, $\epsilon(w) > r$.

Otherwise, $\epsilon(S^j_L) > r$ and $\epsilon(S^j_R) > r$. In this case we halt the induction and let $S = S^j$. Every $w \in L_W(S^{j}) \cup S^j$ satisfies $\epsilon(w) > r$, so $S$ has the properties desired.

At each step, the size of the set $S^j$ halves, so there are at most $\log |W| \leq \log n$ iterations. In each iteration, we perform a constant number of Dijkstras, so the runtime is $O(m \log^2 n)$. %\mina{why $log^2$?}. \jenny{log n iteractions and each dijkstra is $m\log n$ time}
\end{proof}

For each $i$ such that $\epsilon(W_i) \leq r$, let $S_i$ be the subset constructed by applying Claim \ref{binarysearch} to the set $W=W_i$. For each $w \in L_{W_i}(S_i)$, we report that $\epsilon(w) > r$; this holds by Claim \ref{binarysearch}b. If $S_i$ consists of a single vertex $\{s\}$, we can determine that for any $v \in L(W_i)$, $\dm(v, s) \leq \epsilon(s) \leq r \leq kr$, by Claim \ref{binarysearch}a. Otherwise, $|S_i| > 1$, so we report that $\epsilon(s) > r$ for all $s \in S_i$; this holds by Claim \ref{binarysearch}c.

Using a recursive application of our algorithm to the graph $G_i = G[W_i]$, we can certify, for every vertex $w \in W_i$, that  $\epsilon_{G_i}(w) > r$ or that $\epsilon_{G_i}(w) \leq (k-1)r$. Consider any $w \in R_{W_i}(S_i)$. If we determined that $\epsilon_{G_i}(w) > r$, we report that $\epsilon(w) > r$; this holds since $\epsilon(w) \geq \epsilon_{G_i}(w)$. Otherwise, consider any $v \in L(W_i)$. Since $\epsilon(S_i) \leq r$, there is some $s \in S_i$ such that $\dm(v, s) = d(v, s) \leq r$. Then since $\epsilon_{G_i}(w) \leq (k-1)r$ and since $w$ is to the right of $s$ in the topological ordering, we have $\dm(v, w) \leq d(v, s) + d(s, w) \leq r + (k-1)r = kr$. See Figure \ref{fig:min-ecc-2}.

Thus, our algorithm has certified for each $w \in W_i$ that $\epsilon(w) > r$ or that $\dm(v, w) \leq kr$ for all $v \in L(W_i)$. By a symmetric argument, we can construct the set $S_i'$ obtained by applying Claim \ref{binarysearch} to the graph $G$ with the edges reversed; see Figure \ref{fig:min-ecc-2}. Then as above we can determine for each $w \in W_i$  that $\epsilon(w) > r$ or that $\dm(w, v') \leq kr$ for all $v' \in R(W_i)$. Since $W_i$ is a topologically consecutive set, $V \setminus W_i = L(W_i) \cup R(W_i)$. So for any $w \in W_i$, if we determine that $\dm(w, v) \leq kr$ for all $v \in L(W_i)$ and for all $v \in R(W_i)$ we report that $\epsilon(w) \leq kr$; otherwise we report $\epsilon(w) > r$.

\begin{figure}
    \centering
    \includegraphics[width=\textwidth]{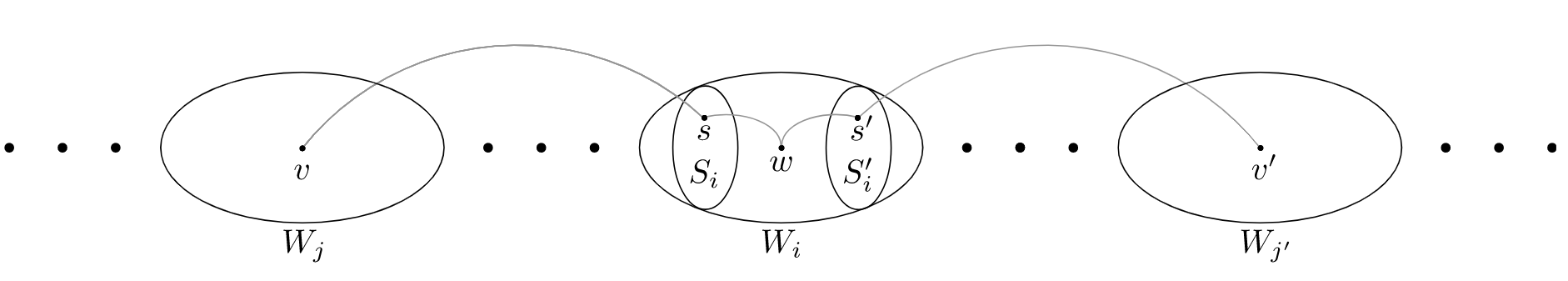}
    \caption{A representation of the $v \to w$ and $w \to v'$ paths, via the sets $S_i$ and $S_i'$ constructed with Claim \ref{binarysearch}. The outer subpaths are of length $ \leq r$, and the inner subpaths are of length $\leq (k-1)r$.}
    \label{fig:min-ecc-2}
\end{figure}

\subsection{Runtime Analysis}\label{min-ecc-appendix}

In this section we analyze the runtime of the algorithm of Proposition \ref{k-prop}, and we give full descriptions of how to prove Theorems \ref{k-approx-1}-\ref{min-radius-2} from Proposition \ref{k-prop} using binary search. 

Recall that $c_k(\tau) = \frac{2^{k-2}(1+\tau)}{2^{k-1}(1+\tau) -\tau}$.

\begin{lemma}\label{k-prop-runtime}
The algorithm of Proposition \ref{k-prop} runs in time $O(\min(mn^{1/k}\log^2 n, m^{2^{k-1}/(2^k -1)}n \log^2 n))$ assuming APSP computations are done in $\tilde{O}(mn)$ time. \\ %\mina{like combinatorial? can you say what is standard way? like $O(mn)$?}
\indent On graphs with integer edge weights bounded by $n^t$ for $t < 3 - \omega$, the algorithm runs in time \\ $O(\min(mn^{1/k}\log^2 n, m^{m^{c_k(\mu(t))}}n \log^2 n))$, assuming APSP computations are done in $O(n^{2+\mu(t)})$ time using Zwick's fast APSP algorithm \cite{zwickbridge}.
\end{lemma}

\begin{proof}
To simultaneously analyze both versions of the algorithm, our algorithm's runtime will be described in terms of a placeholder $\tau$, such that APSP computations within the algorithm are done in $O(n^{2+\tau}\log n)$ time.  To obtain the runtime bound for general weighted DAGs, we will let $\tau = 1$, and note $c_k(1) = \frac{2^{k-1}}{2^k-1}$. To obtain the runtime bound for DAGs with integer edge weights bounded by $n^t$ for $t < 3 - \omega$, we will let $\tau = \mu(t)$.

Topologically sorting the graph takes $O(m \log n)$ time which is absorbed into the final runtime.

In order to use $k=1$ as a base case, our inductive hypothesis will assume a slightly weaker claim about the runtime: in the inductive step for $k$, we will assume there is an $O(\min(mn^{1/(k-1)}\log^2 n, n^{2c_{k-1}(\tau)+1}\log^2 n))$-time  algorithm which certifies for each $v \in V$ that $\epsilon(v) > r$ or that $\epsilon(v) \leq (k-1)r$. Note that $n^{2c_{k-1}}\ge m^{c_{k-1}}$. Then in the base case where $k = 1$, APSP takes time $O(\min(mn\log n, n^{2+\tau}\log n))$, satisfying the inductive hypothesis. %\mina{we have $m$ in the second bound} \jenny{this is the base case, we only need to prove the IH not the actual runtime}

Consider $k > 1$. Running Dijkstra to and from $W_i$ for each $i$ takes $O(mp\log n)$. It takes time $O(mp\log^2 n)$ to apply Claim \ref{binarysearch} twice for each $i$, to construct sets $S_i$ and symmetric sets $S_i'$ (constructed in the same way as the sets $S_i$ but with left and right swapped, pictured in Figure \ref{fig:min-ecc-2}). 

We also do recursive calls of our algorithm on at most $p$ subgraphs, induced by sets $W_i$. Below, we analyze the runtime of the recursive calls in two different ways, giving us two upper bounds on the algorithm's runtime.

\paragraph*{Analysis 1} Let $m_i = |E(G[W_i])|$; then note $\sum_i m_i \leq m$. For each $i$, the recursive call on $W_i$ takes time $O(m_i(n/p)^{1/(k-1)} \log^2 n)$, so in total the recursive calls take time $O(m(n/p)^{1/(k-1)} \log^2 n)$. Let $p = n^{1/k}$, so that $mp = m(n/p)^{1/(k-1)}$. Then the runtime is $O(mn^{1/k} \log^2 n)$.

%fix runtime
\paragraph*{Analysis 2} Since $|W_i| = n/p$, a recursive call on $G[W_i]$ takes time $O((n/p)^{2c_{k-1}(\tau)+1} \log^2 n)$. We do at most $p$ such calls, so the total runtime of the recursive calls is $O((n/p)^{2c_{k-1}(\tau)} n \log^2 n)$. Now, we choose $p$ so that $mp =  (n/p)^{2c_{k-1}(\tau)} n$. Then $m = (n/p)^{2c_{k-1}(\tau) + 1}$. Recall that $c_k(\tau) = \frac{2^{k-2}(1+\tau)}{2^{k-1}(1+\tau) -\tau}$ and note that $2c_{k-1}(\tau) + 1 = \frac{2^{k-1}(1+\tau) - \tau}{2^{k-2}(1+\tau) - \tau} = \frac{2c_{k-1}(\tau)}{c_k(\tau)}$. Thus, $m^{c_k(\tau)} = (n/p)^{2c_{k-1}(\tau)} $. So the runtime of the algorithm is $O((n/p)^{2c_{k-1}(\tau)} \cdot n \log^2 n) = O(m^{c_k(\tau)}n \log^2 n)$. Since $m = O(n^2)$, this satisfies the inductive hypothesis.
\end{proof}

%\subsubsection{Min-Eccentricities Binary Search}\label{min-ecc-binary-search}
%\subsubsection{Proof of Theorem \ref{k-approx-1} and Theorem \ref{k-approx-2}}
%\begin{reminder}{Theorem \ref{k-approx-1}}
%Let $k \geq 2$ be an integer. For any $\delta > 0$, there is an $\tilde{O}(\min(mn^{1/k}/\delta,$ $m^{2^{k-1}/(2^k -1)}n/\delta))$-time algorithm which, given a DAG $G$, outputs for every vertex $v \in V$ an estimate $\epsilon'(v)$ such that $\epsilon(v) \leq \epsilon'(v) < (k+\delta)\epsilon(v)$.
%\end{reminder}

%\begin{reminder}{Theorem \ref{k-approx-2}}
% Let $k \geq 2$ be an integer. For any $\delta > 0$, there is an $\tilde{O}(\min(mn^{1/k}/\delta,$ $m^{c_k(\mu(t))}n/\delta))$-time algorithm which, given a DAG $G$ with integer edge weights bounded by $n^t$ for $t < 3 -\omega$, outputs for every vertex $v \in V$ an estimate $\epsilon'(v)$ such that $\epsilon(v) \leq \epsilon'(v) < (k+\delta)\epsilon(v)$.
%\end{reminder}

%\subsubsection{Min-Radius Binary Search}\label{min-rad-binary-search}

\subsection{Lower Bounds}
In this section, using an essentially linear time reduction, we reduce Triangle Detection to $(2-\delta)$-approximation of min-radius.  %\jenny{should there be tildes in any of these runtimes?} %\jenny{have we defined truly subcubic anywhere?}\mina{added to the prelim} 

\begin{reminder}{Theorem \ref{thm:minradlb}}
If there is a $T(n,m)$-time algorithm for $(2-\delta)$-approximation of min-radius in $O(n)$-node $\tilde{O}(m)$-edge DAGs for some $\delta>0$, then there is an $\tilde{O}(T(n,m) + m)$-time algorithm for Triangle Detection on graphs with $n$ nodes and $m$ edges.
\end{reminder}% \jenny{should this be $2-\delta$?}mina: yup

\begin{proof}
We are going to use two gadgets from previous works: %\jenny{added bullets for legibility, you can get rid of them if you think it's better without} mina: they are nice thanks
\begin{itemize}
\item DAG gadget \cite{avw}: Given a set $X$ of $n$ nodes $v_1,\ldots,v_n$ and a constant integer parameter $t\ge 2$, the gadget creates a DAG $DG_t(X)$ with at most $O(n)$ nodes and $O(n\log{n})$ edges such that in the topological order
of $DG_t(X)$, $v_i<v_{i+1}$, and for any two nodes of $DG_t(X)$ $x, y$ where $x < y$ in the topological order,
$d(x, y) \le t+1$.%\jenny{added the word constant here}

\item Connectivity gadget \cite{sparsereductions}: Let $X=\{v_1,\ldots,v_n\}$, and let $X'=\{v_1',\ldots,v_n'\}$ be a copy of $X$, where both $X$ and $X'$ are independent sets. Then we can add a connectivity gadget $U(X)$ along with edges from $X$ to $U(X)$ and from $U(X)$ to $X'$, such that $|U(X)|=O(\log{n})$, for all $i\neq j$ we have $d(v_i,v_j')=2$, and there is no path from $v_i$ to $v_i'$. 
\end{itemize}

Now let $G=(A,B,C,E_G)$ be an instance of Triangle Detection, with $n$ nodes and $m$ edges. We create a DAG $G^*$ such that if $G$ has a triangle (YES case), the min-radius of $G^*$ is $t+1$, and if $G$ doesn't have a triangle (NO case), the min-radius of $G^*$ is $2t$. We let $t$ be an integer such that $2-\delta/2<\frac{2t}{t+1}$, so that a fast $(2-\delta)$-approximation algorithm is also a fast $(\frac{2t}{t+1}-\delta/2)$-approximation algorithm, and hence it can distinguish min-diameter $t+1$ vs $2t$. % and as a result distinguish the YES case and the NO case of Triangle Detection. 

We define $G^*$ as follows: $G^*$ has $A$, $B$, and $C$ as part of its vertex set. Let $A_1', A_2',\ldots, A_{t+1}'$ be copies of $A$. %\jenny{we should specify what t is (e.g. let t be a positive integer such that $2 - \delta < \frac{2t}{t+1}$...)}.
Add $E_G(A,B)$ to $G^*$ with edges directed from $A$ to $B$, and add $E_G(B,C)$ with edges directed from $B$ to $C$. For any $c\in C$ and $a\in A$, add an edge from $c$ to $a'\in A_2'$ if $a$ and $c$ are attached in $G$, where $a'$ is the copy of $a$ in $A_2'$. For each $i=1,\ldots,t$, connect the copy of $a$ in $A_i'$ to the copy of $a$ in $A_{i+1}'$ for all $a\in A$.

Now we add the two gadgets. Add the connectivity gadget $U(A)$ between $A$ and $A_1'$. 
Add two copies of $DG_t(A)$ sharing $A$, and denote the union of these copies by $DAG(A)$. Also add a node $y$, and add edges from all nodes in $A$ to $y$; this guarantees that the center of $G^*$ must be in $DAG(A)$.

To make all nodes in $A$ at distance $t+1$ to $A_1'$, make $t-1$ copies of $U(A)$, $U_1,\ldots, U_{t-1}$. For each $i=1,\ldots,t-1$, connect the copy of $u$ in $U_i$ to the copy of $u$ in $U_{i+1}$, for any $u\in U(A)$, where $U_t=U(A)$.   
Add edges from all nodes in $A\cup B\cup C$ to all nodes in $U_1$. 

To make all nodes in $A$ at distance $t+1$ to $B$ and $C$, let $x_1,\ldots,x_t$ be a path of length $t-1$. Connect all nodes of $A$ to $x_1$, and connect $x_t$ to all nodes of $B\cup C$. See Figure \ref{fig:minradlb} for the construction. Note that $G^*$ is a DAG, with the order of sets of vertices being $DAG(A),y,x_1,\ldots,x_t,B,C,U_1,\ldots,U_{t-1},U(A),A_1',\ldots,A_{t+1}'$. %\jenny{this list excludes DAG(A), y, and the $x_i$s; I think it'd be useful to include them} mina:done
Moreover, $G^\ast[A \cup B \cup C \cup A_2]$ has $m$ edges corresponding to the original edges of $G^\ast$, and besides those we only added $O(n\log n)$ edges to $G^\ast$. So $G^*$ has $O(n)$ nodes and $O(m+n\log n)$ edges. %\jenny{rephased a bit to avoid implying that the $[C, A_2]$ edges in $G^\ast$ were originally in G, which is sort of almost but not quite true} mina: sounds good.
%Moreover, we only added $\tilde{O}(n)$ edges to the original edges of $G$, so if $G$ has $n$ nodes and $m$ edges, $G^*$ has $O(n)$ nodes and $\tilde{O}(m)$ edges

\begin{figure}[h]
  \centering
  \includegraphics[width=0.7\linewidth]{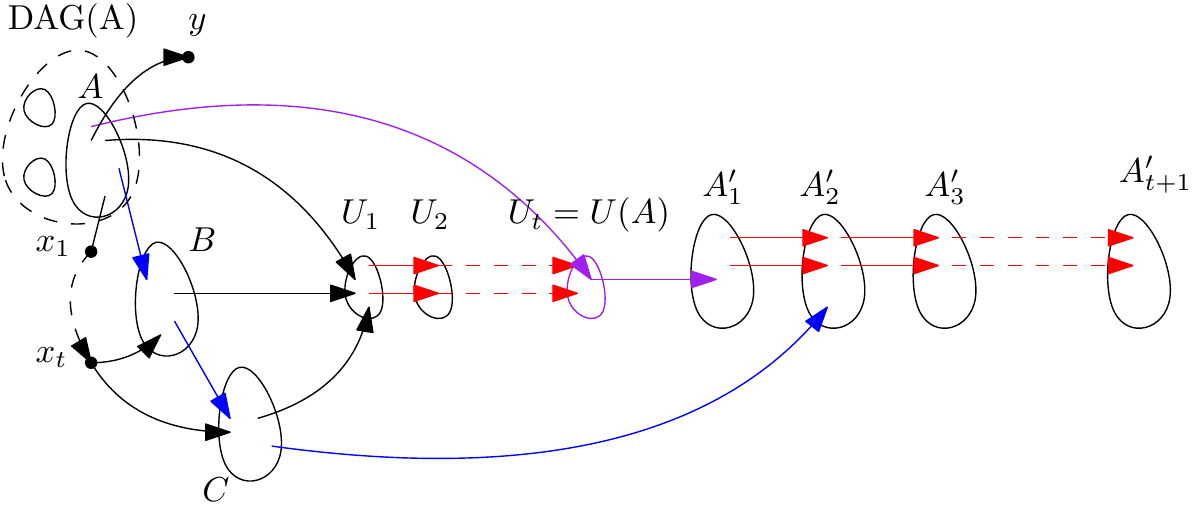}
  \caption{Graph $G^*$ created from the Triangle Detection instance $G$. Blue edges are edges in $G$, red edges are between two nodes that are copies of the same vertex. Purple edges are part of the connectivity gadget. Dashed lines are subpaths.}
  \label{fig:minradlb}
\end{figure}

We will show that if the Triangle Detection instance is a YES instance, then there is a node $a \in A$ such that $\epsilon(a) = t+1$. If the Triangle Detection instance is a NO instance, then we show that for all nodes in $G^*$, their min-eccentricity is at least $2t$.

\paragraph*{YES case.} Let $abc$ be a triangle in $G$. We show that $\epsilon(a) = t+1$. Note that $\dm(a,\bar{a})\le t+1$ for all $\bar{a}\in DAG(A)$. We already know that $d(a,s)\le t+1$ for any $s\in B\cup C\cup \{x_1,\ldots,x_t, y\}$. For any $u \in U_i$ for $i \leq t+1 $, $d(a, u) \leq t+1$ using the path going through $U_1, \dots U_{i-1}$.
%\jenny{could we change this to something like: For any $u \in U_i$ for $i \leq t + 1$, $d(a, u) \leq t$ using the path going through $U_1, \dots U_{i-1}$...currently there's no mention of any $U_i$s except $U(A)$.} 
%For any $u\in U(A)$, $d(a,u)\le t$ using the path going through $U_1,\ldots, U_{t-1}$. 
Since for any $z'\in A_1'$, there is a $u\in U(A)$ that has an edge to $z'$, we have $d(a,z')\le t+1$. Now for all $z'\in A_2'$ where $z'$ is a copy of $z\in A$ and $z\neq a$, we have $d(a,z')=3$ through $U(A)$ and $A_1'$ (using the edges of the connectivity gadget). For $z=a$, using the triangle edges going from $A$ to $B$ to $C$, we have that $d(a,z')=3$. So for all $z'\in A_2'\cup \ldots \cup A_{t+1}'$, we have $d(a,z')\le t+1$.

\paragraph*{NO case.}
Suppose that there is no triangle in $G$. First, note that the min-eccentricities of the vertices outside $DAG(A)$ are infinite, because there is no path between them and $y$. Moreover, if $z\in DAG(A)\setminus A$, it has a copy $z'\in DAG(A)\setminus A$ (in the other copy of $DG_t(A)$), and there is no path between $z$ and $z'$. This is because this path must go through $A$, and since $DAG(A)$ consists of two copies of $DG_t(A)$ sharing $A$, the set of nodes in $A$ that $z$ has a path to (from) is exactly the same as the set of nodes in $A$ that $z'$ has a path to (from). So there is no $a \in A$ such that that $z$ has a path to $a$ and $z'$ has a path from $a$.

Now it remains to compute the min-eccentricities of the vertices in $A$. Let $a\in A$, and let $a_{t+1}'\in A_{t+1}'$ be the copy of $a$. We show that $d(a,a_{t+1}')= 2t$. Let $P$ be a shortest path from $a$ to $a_{t+1}'$. First note that any path from $a$ to $a_{t+1}'$ must go through $a_2'\in A_2'$, where $a_2'$ is a copy of $a$, and we have $d(a_2',a_{t+1}')=t-1$. 
We also know that there is no path from $a$ to $a_2'$ using the edges from $A$ to $U(A)$, because this path would need to contain a path between $a$ and $a_1'\in A_1$ in $G^*[A\cup U(A)\cup A_1']$, and from the construction of the connectivity gadget there is no such path. If $P$ does not use any $C\times A_2'$ edge, then the path must go through $U_i$ for all $i$, and hence it is of length $2t$. So if the min-eccentricity of $a$ is smaller than $2t$, the path $P$ uses a $C\times A_2'$ edge $ca_2'$ for some $c\in C$. If $x_1$ is on the $ac$ path, then the path goes through $x_i$ for all $i$, and hence it is of length $2t$. Then $x_1$ is not on the path, so the $ac$ path must go through $B$. In particular, there is a $b\in B$ such that $ab,bc\in E(G^*)$. Since $ca_2'\in E(G^*)$, this implies that $abc$ is a triangle in $G$, which is a contradiction. So $\epsilon(a) \geq 2t$. 

\end{proof}

\section{Min-diameter}

Our min-diameter approximation algorithm relies on Yuster and Zwick's fast sparse matrix multiplication algorithm. Here, we define $\alpha = \max\{0 \leq r \leq 1 \ | \ \omega(1,r,1) = 2\}$ and  $\beta = \frac{\omega - 2}{1 - \alpha}$.

\begin{theorem}[\cite{yzmm}]\label{yzmm}
If $A$ and $B$ are $n$ by $n$ matrices with at most $l$ nonzero entries each, then $A$ and $B$ can be multiplied in $O(l^{\frac{2\beta}{\beta+1}}n^{\frac{2-\alpha\beta}{\beta+1}} + n^{2 + o(1)})$ time.\footnote{To be precise, given known bounds $\alpha \geq a, \omega \leq c$, one can define $b = \frac{c - 2}{1 - a}$, and then equivalents of Theorem \ref{yzmm} hold for any such pair of values $a, b$, not just for the ``true'' values $\alpha, \beta$. This is implicit in \cite{yzmm}.}
\end{theorem}

This sparse matrix multiplication algorithm will be used to prove the following proposition.

\begin{proposition}\label{mindiamprop}
There is an $O(m^{\frac{2\beta}{3\beta+1}}n^{\frac{4\beta+2-\alpha\beta}{3\beta+1} + o(1)} + n^{2 + o(1)})$-time algorithm which, given an unweighted DAG $G$ and a parameter $D'$,  reports that the min-diameter $D$ of $G$ satisfies $D \leq \left\lceil \frac{3D'}{2} \right\rceil$ or that it satisfies $D > D'$.
\end{proposition}

The algorithm of Proposition \ref{mindiamprop} will be described and proven to work in subsection 3.1, and its runtime will be analyzed in Lemma \ref{mindiamruntime} in subsection \ref{min-diam-appendix}. Then Proposition \ref{mindiamprop} allows us to obtain the min-diameter approximation algorithm given in Theorem \ref{mindiamthm} below. 

\begin{theorem}\label{mindiamthm}
There is an $O(m^{\frac{2\beta}{3\beta+1}}n^{\frac{4\beta+2-\alpha\beta}{3\beta+1} + o(1)} + n^{2 + o(1)})$-time algorithm which, given an unweighted DAG $G$, outputs an estimate $D_0$ for its min-diameter $D$ such that $D \leq D_0 \leq \left\lceil \frac{3D}{2} \right\rceil$.
\end{theorem}

%TO BE CONTINUED - 2022

\begin{proof}
To obtain our approximation $D_0$, we binary search over $D'$ in $[0, n]$ by applying the algorithm of Proposition \ref{mindiamprop} logarithmically many times; note that polylogarithmic factors are $n^{o(1)}$ so they do not affect the runtime bound. Let $C$ be the smallest value found in the binary search such that the algorithm reports that $D \leq \left\lceil \frac{3C}{2}\right\rceil$; then $D > C-1$. Let $D_0 = \left\lceil \frac{3C}{2}\right\rceil$. Then $D \leq D_0 \leq \left\lceil\frac{3D}{2} \right\rceil$, as desired.
\end{proof}

Note that since $\alpha > 0.31389$ \cite{lu18} and $\omega < 2.37286$ \cite{matrixmult2020}, we can use $\beta \simeq 0.5435$. This gives the runtime of $O(m^{0.414}n^{1.522} + n^{2+o(1)})$ stated in Theorem \ref{mindiam}. %\jenny{should this go here or in the introduction?}\mina{here is fine, although it looks different than thm4.} \jenny{oh, yeah I'll update that to the more precise decimal version}

\subsection{Algorithm Description and Correctness}

Our algorithm takes as input an unweighted DAG $G$, an integer $D'$, and a parameter $\epsilon \in [0, 1]$, and reports that $D > D'$ or that $D \leq \left\lceil\frac{3D'}{2} \right\rceil $. (The runtime-minimizing value of $\epsilon$ will be determined later.) 

If at any point, a BFS finds a pair of vertices at min-distance more than $D'$, the algorithm reports that $D > D'$; hence in what follows we will assume that this does not occur. We initially have all pairs of vertices ``unmarked,'' and mark the pairs for which we know that there is a path from one to the other of length at most $\left\lceil\frac{3D'}{2} \right\rceil$. 

The algorithm first takes two preliminary steps: it topologically sorts the graph, and it constructs for each vertex two topologically sorted lists, one of its in-neighbors and one of its out-neighbors.

Our algorithm will then use the greedy set cover algorithm, described in the following lemma. This lemma, and a related randomized version, are standard techniques used in graph distance algorithms (see for example \cite{aingworth, RV13,ChechikLRSTW14,avw}). A proof may be found in \cite{hittingset}. 

\begin{lemma}\label{hslemma}%\mina{changed to lemma}
Let $|V| = n$, let $p = O(n)$, and let $X_1, \dots X_p \subseteq V$ be sets of size $|X_i| \geq n^\epsilon$ for $\epsilon \in [0, 1]$. Then there is an $O(n^{1+\epsilon})$-time algorithm which constructs a set $S \subseteq V$ of size $\tilde{O}(n^{1-\epsilon})$ such that $S \cap X_i \neq \varnothing$ for all $i$.
\end{lemma}

For any $u \in V$, if $|N_{D'/2}^{\text{out}}(u)| < n^\epsilon$ let $X_u = N_{D'/2}^{\text{out}}(u)$ and otherwise let $X_u$ be the left-most $n^\epsilon$ vertices in $N_{D'/2}^{\text{out}}(u)$. So in particular, $|X_u| \leq n^\epsilon$. We can compute $X_u$ as follows: we maintain a list of the $\leq n^\epsilon$ left-most vertices we have found so far that are at distance $< D'/2$ from $u$. At each step, for each vertex in the list, we consider its left-most out-neighbor that is not yet in our set; we add the left-most such out-neighbor to the set. We halt when there are no more such out-neighbors not in our set, or after adding $n^\epsilon$ vertices to our set. Likewise, for any $w\in V$, let $Y_w =  N_{\lceil D'/2 \rceil}^{\text{in}}(w)$ if $|N_{\lceil D'/2} \rceil^{\text{in}}(w)| < n^\epsilon$, and otherwise let $Y_w$ consist of the right-most $n^\epsilon$ vertices in $N_{\lceil D'/2\rceil}^{\text{in}}(w)$. We can compute the sets $Y_w$ in a manner symmetric to how we computed the sets $X_u$. Then we can use Lemma \ref{hslemma} to construct a set $S$ of size $\tilde{O}(n^{1-\epsilon})$ such that for all $u$ having $|N_{D'/2}^{\text{out}}(u)| \geq n^\epsilon$, $S \cap X_u$ is nonempty, and for all $w$ having $|N_{\lceil D'/2\rceil }^{\text{in}}(w)| \geq n^\epsilon$, $S \cap Y_w$ is nonempty.
 
 %\mina{we should have the hitting set lemma, and for that we need to define the sets that we want it to hit, which are the first topologically ordered nodes in the D'/2 neighborhood of every vertex}
Run BFS into and out of every $s \in S$. We may assume that $\dm (s, x) \leq D'$ for all $s \in S, x \in V$.

We will construct matrices $A$ and $B$ with rows and columns indexed by vertices in $V$, as follows: For each vertex $t \in X_u$, let $A[u,t] = 1$. For each vertex $t \in Y_w$, let $B[t, w] = 1$. Multiply $A$ and $B$ using the sparse matrix multiplication algorithm of Theorem \ref{yzmm}.

Now, we will consider any pair of vertices $(u, w)$ where $u$ is to the left of $w$, $u \in R(N_{\lceil D'/2 \rceil }^{\text{in}}(w) \cap S)$, and $w \in L(N_{D'/2}^{\text{out}}(u) \cap S)$. We have that if $d(u, w) \leq D'$, then $(A \cdot B)[u, w] > 0$, and if $(A \cdot B)[u,w] > 0$, then $d(u, w) \leq D' + 1$. Indeed, if $d(u, w) \leq D'$, then there is some intermediate vertex $t$ such that $d(u, t) \leq D'/2$ and $d(t, w) \leq \lceil D'/2 \rceil$. Suppose that $t \not\in X_u$. Then since $X_u$ is defined as the left-most $n^\epsilon$ vertices in $N_{D'/2}^{\text{out}}(u)$, this implies that $|N_{D'/2}^{\text{out}}(u)| > n^\epsilon$ and hence that $|X_u| = n^\epsilon$. Then there is some $s \in S \cap X_u$. Since $t \not\in X_u$, $t$ is to the right of all vertices in $X_u$, and in particular $t$ is to the right of $s$. This implies $t \not\in L(N_{D'/2}^{\text{out}}(u) \cap S)$. But since $w \in L(N_{D'/2}^{\text{out}}(u) \cap S)$ and $t$ lies between $u$ and $w$, this is a contradiction. Thus, $t$ must be in $X_u$, and by symmetry, $t$ is in $Y_w$. So $A[u,t] = 1$ and $B[t, w] = 1$, meaning $(A \cdot B)[u,w] > 0$. Likewise, if $(A \cdot B)[u, w] > 0$, then there exists $t \in X_u \cap Y_w$ such that $d(u, t) \leq D'/2$ and $d(t, w) \leq \lceil D'/2 \rceil$, so $d(u, w) \leq D'+1$. Therefore, we will mark all pairs $(u, w)$ such that $(A \cdot B)[u,w] > 0$.

%Now, consider any $u \in V$ and any $w \in R(X_u)$. %X_u is too narrow, could be stuff on its right; we need to define supersets X_u' and stuff.
%If such a $w$ exists, then it must be the case that $X_u$, then there is some $s \in N_{D'/2}^{\text{out}}(u) \cap S$ such that $s$ is to the left of or is equal to $w$. By assumption, $d(s, w) \leq D'$, so $d(u, w) \leq d(u, s) + d(s, w) \leq D'/2 + D' = \frac{3D'}{2}$. By a symmetric argument, for any $w \in V$ and any $u \in V \cap T_w$ to the right of $w$, we have that $d(u, w) \leq \frac{3D'}{2}$. The algorithm will therefore mark all pairs of vertices $u, w \in V$ except those for which we have simultaneously that $u \in T_w$ and $w \in S_u$.

%Finally, check whether there exists an unmarked pair $(u, w)$. If so, report that $D > D'$. Otherwise, report that $D \leq \frac{3D'}{2}$.

%actually, the $R[u] part is unnecessary, technically

%do something about marking pairs, maybe

Now, consider any $u \in V$ and any $w \not\in L(N_{D'/2}^{\text{out}}(u) \cap S)$ to the right of $u$. We mark the pair $(u, w)$. If such a $w$ exists, then there is some $s \in N_{D'/2}^{\text{out}}(u) \cap S$ such that $s$ is to the left of or is equal to $w$. By assumption, $d(s, w) \leq D'$, so $d(u, w) \leq d(u, s) + d(s, w) \leq D'/2 + D' = \frac{3D'}{2}$. By a symmetric argument, for any $w \in V$ and any $u \not\in R(N_{\lceil D'/2 \rceil}^{\text{in}}(w) \cap S)$ to the left of $w$, we have that $d(u, w) \leq \left\lceil \frac{3D'}{2} \right\rceil$, so again we mark any such pair $(u, w)$. %\mina{change the next sentence to say: we have marked all pair $(u,w)$ that are $D$ close except ..., which we marked in the last step.}
Thus, since we have assumed that $\epsilon(s) \leq D'$ for all $s \in S$, the algorithm will mark all pairs of vertices $u, w \in V$ except those for which we have simultaneously that $u \in R(N_{\lceil D'/2 \rceil}^{\text{in}}(w) \cap S)$ and $w \in L(N_{D'/2}^{\text{out}}(u) \cap S)$.

Finally, check whether there exists an unmarked pair $(u, w)$. If so, report that $D > D'$. Otherwise, report that $D \leq \left\lceil \frac{3D'}{2} \right\rceil$.

%%%%%%%%%%%%%%%%%%%%%%%%%%%%%%%%%%

\subsection{Runtime Analysis}\label{min-diam-appendix}

Here we analyze the runtime of the algorithm of Proposition \ref{mindiamprop}.%, and we describe how to prove Theorem \ref{mindiamthm} from Proposition \ref{mindiamprop} using binary search. 

\begin{lemma}\label{mindiamruntime}
The algorithm of Proposition \ref{mindiamprop} runs in time $\tilde{O}(m^{\frac{2\beta}{3\beta+1}}n^{\frac{4\beta+2-\alpha\beta}{3\beta+1} + o(1)} + n^{2 + o(1)})$.
\end{lemma}

\begin{proof}

Topologically sorting the graph takes $O(m \log n)$ time which is absorbed into the final runtime. Constructing for each vertex topologically ordered lists of its in-neighbors and out-neighbors can be done in time $\tilde{O}(n^2)$.

Computing the covering set $S$ takes time $\tilde{O}(n^{1+\epsilon})$ and running BFS from its vertices takes time $O(n^{1-\epsilon}m\log n)$. Checking for each pair $(u, w)$ whether $u \in L(N_{D'/2}^{\text{out}}(u) \cap S)$ and $w \in L(N_{D'/2}^{\text{out}}(u) \cap S)$ can be done in $\tilde{O}(n^2)$ time.

For a fixed $u$, to compute $X_u$, we maintain a list of the at most $n^{\epsilon}$ left-most vertices we have found that are at distance $<D'/2$ from $u$. For each vertex, we store its left-most out-neighbor that is not yet in our set. At each step, we find the left-most such out-neighbor of any vertex in the list; this takes time $O(n^\epsilon)$, and updating the list to reflect that this out-neighbor has been added to our set takes time $O(n^\epsilon)$. At each step we add a vertex to our set $X_u$, so there are at most $O(n^\epsilon)$ steps. Hence, constructing $X_u$ for a fixed $u$ takes $O(n^{2\epsilon})$ time. Then constructing all sets $X_u, Y_w$ takes $O(n^{1+2\epsilon})$ time altogether.

Finally, note that there are at most $n^\epsilon$ 1s in each row of $A$, since we only set $A[u, t] = 1$ if $t \in X_u$. Thus, $A$ contains at most $n^{1+\epsilon}$ 1s. By symmetry, the same holds for $B$. Then multiplying $A$ and $B$ can be done in time $O(n^{(1+\epsilon)\frac{2\beta}{\beta+1} + \frac{2-\alpha\beta}{\beta+1} + o(1)} + n^{2 + o(1)})$, using Yuster and Zwick's fast sparse matrix multiplication (Theorem \ref{yzmm}).

Then the total runtime is:
$$\tilde{O}\left(n^{1-\epsilon}m + n^{1+2\epsilon} + n^{(1+\epsilon)\frac{2\beta}{\beta+1} + \frac{2-\alpha\beta}{\beta+1} + o(1)} + n^{2 + o(1)}\right)$$

Let $\gamma$ be the largest value such that $n^\gamma = O(m)$. Let $\epsilon = \frac{\alpha\beta + (\beta+1)(\gamma - 1)}{3\beta+1}$; this value is chosen because it sets the first and third terms in the above runtime equal (up to $n^{o(1)}$ factors), hence asymptotically minimizing their sum. Substituting the value of $\epsilon$ and simplifying, the runtime of the algorithm is:

$$\tilde{O}\left(n^{\frac{2\beta}{3\beta+1}\gamma+\frac{4\beta+2-\alpha\beta}{3\beta+1} + o(1)} + n^{\frac{2\beta+2}{3\beta+1}\gamma + \frac{\beta - 1 + 2\alpha\beta}{3\beta+1}} + n^{2+o(1)}\right) $$

%comparing exponents
%(4\beta + 2 - \alpha\beta) > (3\beta+1)+2\alpha\beta - 2\beta - 2
%3\beta > 3\alpha\beta - 1
%3\beta(1-\alpha) > -1
%\beta > (-1/3)/(1-\alpha)

We note that $3\beta-3\alpha\beta > 3(\omega - 2) \geq 0 > -1$, giving:

$$4\beta+2-\alpha\beta > 2 + (\beta - 1 + 2\alpha\beta) \geq 2\gamma + (\beta - 1 + 2\alpha\beta)$$

Thus, the first term of the above runtime dominates the second. Substituting $n^\gamma = O(m)$, and noting that the polylogarithmic factors in the runtime are of order $n^{o(1)}$, the runtime is $O(m^{\frac{2\beta}{3\beta+1}}n^{\frac{4\beta+2-\alpha\beta}{3\beta+1} + o(1)} + n^{2+o(1)})$, as desired.

\end{proof}

\section*{Acknowledgements}

We thank our advisor, Virginia Vassilevska Williams, for many helpful suggestions. 

\bibliography{references}

\begin{thebibliography}{10}

\bibitem{abboud2018if}
Amir Abboud, Arturs Backurs, and Virginia~Vassilevska Williams.
\newblock If the current clique algorithms are optimal, so is {Valiant}'s
  parser.
\newblock {\em SIAM Journal on Computing}, 47(6):2527--2555, 2018.

\bibitem{AbboudGW15}
Amir Abboud, Fabrizio Grandoni, and Virginia {Vassilevska Williams}.
\newblock Subcubic equivalences between graph centrality problems, {APSP} and
  diameter.
\newblock In {\em Proceedings of the Twenty-Sixth Annual {ACM-SIAM} Symposium
  on Discrete Algorithms, {SODA} 2015, San Diego, CA, USA, January 4-6, 2015},
  pages 1681--1697, 2015.

\bibitem{avw}
Amir Abboud, Virginia {Vassilevska Williams}, and Joshua~R. Wang.
\newblock Approximation and fixed parameter subquadratic algorithms for radius
  and diameter in sparse graphs.
\newblock In {\em Proceedings of the Twenty-Seventh Annual {ACM-SIAM} Symposium
  on Discrete Algorithms, {SODA} 2016, Arlington, VA, USA, January 10-12,
  2016}, pages 377--391, 2016.

\bibitem{sparsereductions}
Udit Agarwal and Vijaya Ramachandran.
\newblock Fine-grained complexity for sparse graphs.
\newblock In {\em Proceedings of the 50th Annual ACM SIGACT Symposium on Theory
  of Computing}, pages 239--252, 2018.

\bibitem{aingworth}
D.~Aingworth, C.~Chekuri, P.~Indyk, and R.~Motwani.
\newblock Fast estimation of diameter and shortest paths (without matrix
  multiplication).
\newblock {\em SIAM J. Comput.}, 28(4):1167--1181, 1999.

\bibitem{matrixmult2020}
Josh Alman and Virginia~Vassilevska Williams.
\newblock A refined laser method and faster matrix multiplication.
\newblock In {\em Proceedings of the 32nd Annual {ACM-SIAM} Symposium on
  Discrete Algorithms, {SODA} 2021}, 2020.

\bibitem{BenMoshe}
B.~Ben-Moshe, B.~K. Bhattacharya, Q.~Shi, and A.~Tamir.
\newblock Efficient algorithms for center problems in cactus networks.
\newblock {\em Theoretical Computer Science}, 378(3):237 -- 252, 2007.

\bibitem{BeKa07}
P.~Berman and S.~P. Kasiviswanathan.
\newblock Faster approximation of distances in graphs.
\newblock In {\em Proc. WADS}, pages 541--552, 2007.

\bibitem{BCH+15}
Michele Borassi, Pierluigi Crescenzi, Michel Habib, Walter~A. Kosters, Andrea
  Marino, and Frank~W. Takes.
\newblock Fast diameter and radius {BFS}-based computation in (weakly
  connected) real-world graphs: With an application to the six degrees of
  separation games.
\newblock {\em Theoretical Computer Science}, 586:59--80, 2015.

\bibitem{bringmann2018clique}
Karl Bringmann and Philip Wellnitz.
\newblock {Clique-Based Lower Bounds for Parsing Tree-Adjoining Grammars}.
\newblock In Juha K{\"a}rkk{\"a}inen, Jakub Radoszewski, and Wojciech Rytter,
  editors, {\em 28th Annual Symposium on Combinatorial Pattern Matching (CPM
  2017)}, volume~78 of {\em Leibniz International Proceedings in Informatics
  (LIPIcs)}, pages 12:1--12:14, Dagstuhl, Germany, 2017. Schloss
  Dagstuhl--Leibniz-Zentrum fuer Informatik.
\newblock URL: \url{http://drops.dagstuhl.de/opus/volltexte/2017/7332}, \href
  {https://doi.org/10.4230/LIPIcs.CPM.2017.12}
  {\path{doi:10.4230/LIPIcs.CPM.2017.12}}.

\bibitem{seth2}
Chris Calabro, Russell Impagliazzo, and Ramamohan Paturi.
\newblock The complexity of satisfiability of small depth circuits.
\newblock In {\em International Workshop on Parameterized and Exact
  Computation}, pages 75--85. Springer, 2009.

\bibitem{Chan12}
T.~M. Chan.
\newblock All-pairs shortest paths for unweighted undirected graphs in {\it
  o}({\it mn}) time.
\newblock {\em ACM Transactions on Algorithms}, 8(4):34, 2012.

\bibitem{ChechikLRSTW14}
Shiri Chechik, Daniel~H. Larkin, Liam Roditty, Grant Schoenebeck, Robert~Endre
  Tarjan, and Virginia {Vassilevska Williams}.
\newblock Better approximation algorithms for the graph diameter.
\newblock In {\em Proceedings of the Twenty-Fifth Annual {ACM-SIAM} Symposium
  on Discrete Algorithms, {SODA} 2014, Portland, Oregon, USA, January 5-7,
  2014}, pages 1041--1052, 2014.

\bibitem{Chepoi02}
V.~Chepoi, F.~Dragan, and Y.~Vax\`{e}s.
\newblock Center and diameter problems in plane triangulations and
  quadrangulations.
\newblock In {\em Proc. SODA}, pages 346--355, 2002.

\bibitem{ChepoiD94}
V.~Chepoi and F.~F. Dragan.
\newblock A linear-time algorithm for finding a central vertex of a chordal
  graph.
\newblock In {\em ESA}, pages 159--170, 1994.

\bibitem{chung}
F.~R.~K. Chung.
\newblock Diameters of graphs: Old problems and new results.
\newblock {\em Congr. Numer.}, 60:295--317, 1987.

\bibitem{Corneil01}
D.G. Corneil, F.F. Dragan, M.~Habib, and C.~Paul.
\newblock Diameter determination on restricted graph families.
\newblock {\em Discr. Appl. Math.}, 113:143 -- 166, 2001.

\bibitem{CW99}
L.~Cowen and C.~Wagner.
\newblock Compact roundtrip routing for digraphs.
\newblock In {\em SODA}, pages 885--886, 1999.

\bibitem{d19}
Mina Dalirrooyfard, Virginia~Vassilevska Williams, Nikhil Vyas, Nicole Wein,
  Yinzhan Xu, and Yuancheng Yu.
\newblock Approximation algorithms for min-distance problems.
\newblock In {\em 46th International Colloquium on Automata, Languages, and
  Programming (ICALP 2019)}. Schloss Dagstuhl-Leibniz-Zentrum fuer Informatik,
  2019.

\bibitem{Dvir04}
D.~Dvir and G.~Handler.
\newblock The absolute center of a network.
\newblock {\em Networks}, 43:109 -- 118, 2004.

\bibitem{eppstein-planar-jv}
D.~Eppstein.
\newblock Subgraph isomorphism in planar graphs and related problems.
\newblock {\em J. Graph Algorithms and Applications}, 3(3):1--27, 1999.

\bibitem{FHW12}
Silvio Frischknecht, Stephan Holzer, and Roger Wattenhofer.
\newblock Networks cannot compute their diameter in sublinear time.
\newblock In {\em Proceedings of the Twenty-Third Annual ACM-SIAM Symposium on
  Discrete Algorithms}, pages 1150--1162. SIAM, 2012.

\bibitem{lu18}
Fran{\c{c}}ois~Le Gall and Florent Urrutia.
\newblock Improved rectangular matrix multiplication using powers of the
  coppersmith-winograd tensor.
\newblock In {\em Proceedings of the Twenty-Ninth Annual ACM-SIAM Symposium on
  Discrete Algorithms}, pages 1029--1046. SIAM, 2018.

\bibitem{Hakimi}
S.L. Hakimi.
\newblock Optimum location of switching centers and absolute centers and
  medians of a graph.
\newblock {\em Oper. Res.}, 12:450 -- 459, 1964.

\bibitem{ipz1}
R.~Impagliazzo and R.~Paturi.
\newblock On the complexity of k-{SAT}.
\newblock {\em Journal of Computer and System Sciences}, 62(2):367--375, 2001.

\bibitem{Pettie02}
Seth Pettie.
\newblock A faster all-pairs shortest path algorithm for real-weighted sparse
  graphs.
\newblock In {\em International Colloquium on Automata, Languages, and
  Programming}, pages 85--97. Springer, 2002.

\bibitem{pettie2002computing}
Seth Pettie and Vijaya Ramachandran.
\newblock Computing shortest paths with comparisons and additions.
\newblock In {\em Proceedings of the thirteenth annual ACM-SIAM symposium on
  Discrete algorithms}, pages 267--276, 2002.

\bibitem{RV13}
Liam Roditty and Virginia Vassilevska~Williams.
\newblock Fast approximation algorithms for the diameter and radius of sparse
  graphs.
\newblock In {\em Proceedings of the forty-fifth annual ACM symposium on Theory
  of computing}, pages 515--524, 2013.

\bibitem{hittingset}
Virginia Vassilevksa~Williams, Nike Sun, and Nishith Khandwala.
\newblock Lecture notes in graph algorithms (hitting sets, {APSP}), October
  2016.
\newblock URL: \url{http://theory.stanford.edu/~virgi/cs267/lecture5.pdf}.

\bibitem{WeYu13}
O.~Weimann and R.~Yuster.
\newblock Approximating the diameter of planar graphs in near linear time.
\newblock In {\em Proc. ICALP}, 2013.

\bibitem{w05}
Ryan Williams.
\newblock A new algorithm for optimal 2-constraint satisfaction and its
  implications.
\newblock {\em Theoretical Computer Science}, 348(2-3):357--365, 2005.

\bibitem{williams2018subcubic}
Virginia~Vassilevska Williams and R~Ryan Williams.
\newblock Subcubic equivalences between path, matrix, and triangle problems.
\newblock {\em Journal of the ACM (JACM)}, 65(5):1--38, 2018.

\bibitem{WN08}
C.~{Wulff-Nilsen}.
\newblock Wiener index, diameter, and stretch factor of a weighted planar graph
  in subquadratic time.
\newblock {\em Technical report, University of Copenhagen}, 2008.

\bibitem{Yuster10}
Raphael Yuster.
\newblock Computing the diameter polynomially faster than {APSP}.
\newblock {\em arXiv preprint arXiv:1011.6181}, 2010.

\bibitem{yzmm}
Raphael Yuster and Uri Zwick.
\newblock Fast sparse matrix multiplication.
\newblock {\em ACM Trans. Algorithms}, 1(1):2–13, 2005.
\newblock \href {https://doi.org/10.1145/1077464.1077466}
  {\path{doi:10.1145/1077464.1077466}}.

\bibitem{zwickbridge}
U.~Zwick.
\newblock All pairs shortest paths using bridging sets and rectangular matrix
  multiplication.
\newblock {\em J. ACM}, 49(3):289--317, 2002.

\end{thebibliography}

%\appendix

\end{document}